\documentclass{article}
\usepackage{amsthm}
\usepackage{amsmath}
\usepackage{amsfonts}
\usepackage{latexsym}
\usepackage{graphicx}
\usepackage{moreverb}

\newtheorem{definition}{Definition}

\newtheorem{lemma}{Lemma}
\newtheorem{theorem}{Theorem}

\newcommand{\ra}{\rightarrow}
\newcommand{\bs}[1]{\boldsymbol{#1}}
\newcommand{\ov}[1]{\overline{#1}}

\newcommand{\lss}{\langle}
\newcommand{\rss}{\rangle}

\newcommand{\MSE}{\mathrm{MSE}}
\newcommand{\MISE}{\mathrm{MISE}}

\newcommand{\N}{\mathbb{N}}
\newcommand{\R}{\mathbb{R}}
\newcommand{\C}{\mathbb{C}}
\newcommand{\E}{\mathbb{E}}
\newcommand{\A}{\mathcal{A}}
\newcommand{\B}{\mathcal{B}}
\newcommand{\F}{\mathcal{F}}
\newcommand{\bsm}[1]{\bs{\mathrm{#1}}}

\begin{document}
%
\title{Kernel density estimates in particle filter}
%
%
%

\author{David~Coufal
\thanks{D. Coufal is with the Department
of Nonlinear Dynamics and Complex Systems of the Institute of Computer Science AS CR, 
Pod Vod\'{a}renskou v\v{e}\v{z}\'{\i}~2, 182 07 Prague 8,
Czech Republic. E-mail: david.coufal@cs.cas.cz.}}

%
%

\markboth{Journal of \LaTeX\ Class Files,~Vol.~XX, No.~X, February~2014 (14.2.2014)}%
{Shell \MakeLowercase{\textit{et al.}}: Bare Demo of IEEEtran.cls for Journals}
%



\maketitle

\begin{abstract}
The paper deals with kernel density estimates of~filtering densities
in the particle filter. The convergence of the estimates is investigated
by means of Fourier analysis. It is shown that the estimates converge 
to the theoretical filtering densities in the mean integrated squared error
under a certain assumption on the Sobolev character of the filtering 
densities. A~sufficient condition is presented for the persistence of 
this Sobolev character over time. Both results are extended to
partial derivatives of the estimates and filtering densities.
\end{abstract}


%

\section{Introduction}
%
%
%
%
The particle filter enables its user to efficiently compute integral
characteristics (moments) of distributions of interest. In the filtering problem, these
distributions are traditionally referred to as the \textit{filtering distributions}.
In the particle filter, the filtering distribution is approximated by an empirical measure.
This measure is implemented in the form of a weighted sum of Dirac measures
located at randomly (empirically) generated points called \textit{particles}.
Particles are generated sequentially by the algorithm which is an instance of the
\textit{sequential Monte Carlo methods}~\cite{Doucet2001, Doucet2011}.

The theoretical result that justifies the application of the particle filter
is that the generated empirical measures converge to the 
theoretical filtering distribution as the number of particles goes to infinity
\cite{Doucet2001, Crisan2002}. Approximating the filtering distribution by
an empirical measure is extremely useful for estimating moments of
the distribution because they correspond to weighted sums of values
of moment functions over generated particles.

The filtering distribution has typically a density with respect to the corresponding
Lebesgue measure. This density is called the \textit{filtering density}.
The knowledge of a suitable analytical approximation of the filtering density
has several advantages. Let us mention, for example, the possibility of
computing densities of related conditional distributions and conditional expected 
values in an analytical form. The other benefit is that one can get a~deeper
insight into the character of the filtering distribution through the analysis
of its density approximation.

From these practical, and of course also theoretical, reasons the issue
of the analytical approximation of the filtering densities is the subject of 
ongoing research. The problem has been addressed 
in \cite{Doucet2001}, Chapter 12, \cite{Oudjane2004, Kunsch2005}
and recently in \cite{Crisan2013}.

In this paper, we deal with the estimation/approximation of filtering densities
using the nonparametric kernel density estimation methodology. We use an
approach based on Fourier analysis inspired by the book of Tsybakov
\cite{Tsybakov2009}. We will show that the convergence of kernel density 
estimates is assured even if the particles generated by the particle filter
are not~i.i.d., which is the common assumption in the application of kernel
methods.

The paper presents two main results. The first result is the convergence
of the kernel density estimates to the theoretical filtering density at a fixed time of
operation of the filter, provided that the number of generated particles goes to
infinity. The result is based on the notion of the Sobolev character of the 
filtering density. The second result gives a condition under which
this Sobolev character is retained over time. Thus, the first result applies at
any time of operation of the filter. Both results are extended to partial
derivatives of the estimates and filtering densities.

The rest of the paper is organized as follows. In the next section we review
the basics of the particle filter's theory together with the related
convergence results. Section~\ref{SecIII} deals with a review of nonparametric
kernel density estimation methods with the focus on the Fourier analysis approach. 
Sections \ref{SecIV}~and~\ref{SecV} present the announced main results of 
the paper. Section~\ref{SecVI} shows an application of the developed theory
in an example related to the Kalman filter. The paper is concluded by 
Section~\ref{SecVII}.

\section{Particle filter}
\label{SecII}
The basics of the particle filter and general filtering theory can be
found in \cite{Doucet2001, Doucet2011, Crisan2002, Krylov2007} and
\cite{Sarkka2013}. However, there is a plenty of other literature specialized
in these subjects. Nevertheless, we present here the essential framework
of the related methodology in order that the paper be self-contained.

\subsection{Filtering problem}
\label{SigSec}
The filtering problem is the task of determining the optimal
estimate of an inaccessible value of the actual state of a stochastic process
on the basis of knowledge of accessible observations. The observations
establish a stochastic process called the \textit{observation process}.
The observation process is interconnected with a principal stochastic
process which is called the \textit{signal process}. Let us be more specific.

Let $(\Omega,\mathcal{A},P)$ be a probabilistic space with two
stochastic processes $\{\bs{X}_t\}^\infty_{t=0}$, $\{\bs{Y}_t\}^\infty_{t=1}$
specified on it. The first process $\{\bs{X}_t\}^\infty_{t=0}$,
$\bs{X}_t:(\Omega,\A)\ra(\R^{d_x},\B(\R^{d_x}))$, $t\in\N_0$, $d_x\in\N$
is the signal process. The signal process is considered to represent
generally an inhomogeneous Markov chain with a continuous state space.
The probabilistic behavior of the chain is determined by the 
initial distribution $\pi_0(d\bs{x}_0)$ of $\bs{X}_0$
and by the set of transition kernels $K_{t-1}:\B(\R^{d_x})\times\R^{d_x}\ra [0,1]$,
$t\in\N$. We denote by $K_{t-1}(d\bs{x}_{t}|\bs{x}_{t-1})$ 
the measure represented by the transition kernel $K_{t-1}$ for
$\bs{x}_{t-1}\in\R^{d_x}$ being fixed.

Let $\{\bs{Y}_t\}^\infty_{t=1}$, 
$\bs{Y}_t:(\Omega,\A)\!\ra\!(\R^{d_y}, \B(\R^{d_y}))$, $t\in\N$, $d_y\in\N$
be the observation process specified on the basis of the signal process by
formula
\begin{equation}
\label{YtDef}
\bs{Y}_t=h_t(\bs{X}_t)+\bs{V}_t,\;\;t\in\N,
\end{equation}
where $h_t:\R^{d_x}\ra\R^{d_y}$, $t\in\N$ are Borel functions and
$\bs{V}_t$ are (all)-other-variables independent random variables specified on
$(\Omega,\mathcal{A},P)$. That is, 
$\bs{V}_t:(\Omega,\mathcal{A})\ra(\R^{d_y},\B(\R^{d_y}))$, 
$t\in\N$, $d_y\in\N$ and 
$P(\bs{V}_t\in d\bs{v}_t|\bs{X}_{0:t},\bs{Y}_{1:t-1},\bs{V}_{1:t-1})=
P(\bs{V}_t\in d\bs{v}_t)$ for all $t\in\N$.
The (all)-other-variables independence of $\bs{V}_t$
transfers on observations in the following way:
\begin{equation}
\label{YtCondInd}
P(\bs{Y}_t\in d\bs{y}_t|\bs{X}_{0:t},\bs{Y}_{1:t-1})=
P(\bs{Y}_t\in d\bs{y}_t|\bs{X}_t).
\end{equation}
Indeed, we have $\sigma(\bs{X}_{0:t},\bs{Y}_{1:t-1})=
\sigma(\bs{X}_{0:t},\bs{V}_{1:t-1})$ due to (\ref{YtDef}). 
$\bs{V}_{1:t-1}$ is independent of $(\bs{Y}_t,\bs{X}_{0:t})$,
therefore 
$P(\bs{Y}_t\in d\bs{y}_t|\bs{X}_{0:t},\bs{V}_{1:t-1})=
P(\bs{Y}_t\in d\bs{y}_t|\bs{X}_{0:t})$. The assertion is finally obtained
by the Markov property of the signal process. Remark that for $t\!=\!1$,
the left-hand side of (\ref{YtCondInd}) reads as 
$P(\bs{Y}_1\in d\bs{y}_1|\bs{X}_{0:1})$.

\subsection{Filtering distribution and filtering density}
\label{ftds}
As stated, the purpose of filtering is to present the optimal estimate of 
the actual state $\bs{x}_t\in\R^{d_x}$ of the signal process using
the actual and past observations $\bs{y}_{1:t}=(\bs{y}_1,\dots,\bs{y}_t)$. 
This is done at each time instant $t\in\N$. It is the classical result that
under the assumption of $L_2$ integrability of $\bs{X}_t$,
the $L_2$-optimal estimate corresponds to the conditional expectation
$\E[\bs{X}_t|\bs{Y}_{1:t}]$. In what follows we will assume that 
$\bs{X}_t\in L_2(\Omega,\A,P)$ for each $t\in \N_0$.

For fixed observations $\bs{Y}_{1:t}=\bs{y}_{1:t}$, 
the conditional expectation $\E[\bs{X}_t|\bs{Y}_{1:t}=\bs{y}_{1:t}]$
can be determined on the basis of the related
conditional distribution $P(\bs{X}_t\in d\bs{x}_t|\bs{Y}_{1:t}=\bs{y}_{1:t})$.
This distribution then represents the filtering distribution at time \mbox{$t\in\N$} 
and will be approximated by an empirical measure generated by 
the particle filter.

In the standard setting of the filtering problem, all the involved
finite-dimensional distributions have bounded and continuous densities 
with respect to the corresponding Lebesgue measures. Especially, 
we assume that $\pi_0(d\bs{x}_0)=p_0(\bs{x}_0)\,d\bs{x}_0$, 
$K_{t-1}(d\bs{x}_t|\bs{x}_{t-1})=K_{t-1}(\bs{x}_{t}|\bs{x}_{t-1})\,d\bs{x}_{t}$
and $P(\bs{V}_t\in d\bs{v}_t)=g^v_t(\bs{v}_t)\,d\bs{v}_t$.
This enables us to identify the respective filtering density, which is the density of
$P(\bs{X}_t\in d\bs{x}_t|\bs{Y}_{1:t}=\bs{y}_{1:t})$.

The conditional density of \mbox{$P(\bs{Y}_t\in d\bs{y}_t|\bs{X}_t=\bs{x}_t)$} 
is determined by formula (\ref{YtDef}). The density is denoted 
$g_t(\bs{y}_t|\bs{x}_t)$
and writes as
\begin{equation}
\label{gtdef}
g_t(\bs{y}_t|\bs{x}_t)=g^v_t(\bs{y}_t-h_t(\bs{x}_t)).
\end{equation}
The joint density of $(\bs{X}_{0:t},\bs{Y}_{1:t})$ has then form
\begin{equation}
\label{XtYtJoint}
p(\bs{x}_{0:t},\bs{y}_{1:t})=p_0(\bs{x}_0)
\prod_{k=1}^t g_k(\bs{y}_k|\bs{x}_k)\,K_{k-1}(\bs{x}_k|\bs{x}_{k-1}).
\end{equation}
These specifications are induced by the conditional independence of
observations (\ref{YtCondInd}) and by the standard theory of 
Markov chains with a continuous state space.

The filtering density is $p(\bs{x}_t|\bs{y}_{1:t})$ for \mbox{$t\in\N$.}
Employing the joint distribution (\ref{XtYtJoint}), we have
\begin{equation}
p(\bs{x}_t|\bs{y}_{1:t})=
\frac{p(\bs{x}_t,\bs{y}_{1:t})}{p(\bs{y}_{1:t})}=
\frac{\int p(\bs{x}_{0:t},\bs{y}_{1:t})\,d\bs{x}_{0:t-1}}
{\int p(\bs{x}_{0:t},\bs{y}_{1:t})\,d\bs{x}_{0:t}}.
\end{equation}

The above integrals are generally inexpressible in a~closed form.
However, certain recursive analytical relations can be stated.
These relations are called the \textit{filtering equations} and 
are addressed in the next section.

\subsection{Filtering equations}
\label{FeSec}
The filtering equations describe recursively the evolution of the filtering
density $p(\bs{x}_t|\bs{y}_{1:t})$ over time. They consists of the 
\textit{prediction formula} (\ref{L1}) and the \textit{update formula}~(\ref{L2}).

\begin{lemma}
Let the joint density be given by formula \textnormal{(\ref{XtYtJoint})}, 
then
\begin{equation}
\label{L1}
p(\bs{x}_t|\bs{y}_{1:t-1})=
\int K_{t-1}(\bs{x}_t|\bs{x}_{t-1})p(\bs{x}_{t-1}|\bs{y}_{1:t-1})\,d\bs{x}_{t-1}
\end{equation}
for $t\geq 2$, and 
$p(\bs{x}_1)=\int K_{0}(\bs{x}_1|\bs{x}_{0})p_0(\bs{x}_{0})\,d\bs{x}_{0}$
for $t=1$.
\end{lemma}

\noindent
\textbf{Proof.} We get the result from (\ref{XtYtJoint}) by series of
integrations. Let us start with $t=1$. In this case,
formula~(\ref{XtYtJoint}) reads as
$p(\bs{x}_{0:1},\bs{y}_1)\!=\!g_1(\bs{y}_1|\bs{x}_1)
K_{0}(\bs{x}_1|\bs{x}_0)\,p_0(\bs{x}_0)$.
By integrating out $\bs{y}_1$ we get 
$p(\bs{x}_{0:1})\!=\!K_{0}(\bs{x}_1|\bs{x}_0)\,p_0(\bs{x}_0)$ and
the result is obtained by integration with respect to $\bs{x}_0$.

In the general case of $t\geq 2$, we get the following expres\-sions by
the transcription of (\ref{XtYtJoint}) and integrating out $\bs{y}_t$,
\begin{eqnarray*}
p(\bs{x}_{0:t},\bs{y}_{1:t})&\!\!\!=\!\!\!&
g_t(\bs{y}_t|\bs{x}_t)K_{t-1}(\bs{x}_t|\bs{x}_{t-1})p(\bs{x}_{0:t-1},\bs{y}_{1:t-1}),\\
p(\bs{x}_{0:t},\bs{y}_{1:t-1})&\!\!\!=\!\!\!&
K_{t-1}(\bs{x}_t|\bs{x}_{t-1})p(\bs{x}_{0:t-1},\bs{y}_{1:t-1}).
\end{eqnarray*}
Subsequently, the integration w.r.t. $\bs{x}_{0:t-2}$ and $\bs{x}_{t-1}$ gives
\begin{eqnarray*}
p(\bs{x}_{t-1:t},\bs{y}_{1:t-1})&\!\!\!=\!\!\!&
K_{t-1}(\bs{x}_t|\bs{x}_{t-1})p(\bs{x}_{t-1},\bs{y}_{1:t-1}),\\
p(\bs{x}_{t},\bs{y}_{1:t-1})&\!\!\!=\!\!\!&
\!\!\int\!\! 
K_{t-1}(\bs{x}_t|\bs{x}_{t-1})p(\bs{x}_{t-1},\bs{y}_{1:t-1})\,d\bs{x}_{t-1}.
\end{eqnarray*}
Finally, dividing both sides of the last formula by the marginal density
$p(\bs{y}_{1:t-1})$ gives the result.\hfill$\Box$

\begin{lemma}
Let the joint density be given by formula \textnormal{(\ref{XtYtJoint})}, 
then
\begin{equation}
\label{L2}
p(\bs{x}_{t}|\bs{y}_{1:t})=
\frac{g_t(\bs{y}_t|\bs{x}_t)p(\bs{x}_{t}|\bs{y}_{1:t-1})}
{\int g_t(\bs{y}_t|\bs{x}_t)p(\bs{x}_{t}|\bs{y}_{1:t-1})\,d\bs{x}_{t}},\;\;\;t\in\N,
\end{equation}
with $p(\bs{x}_{1}|\bs{y}_{1:0})$ understood as $p(\bs{x}_{1})$ for $t=1$.
\end{lemma}

\noindent
\textbf{Proof.}
We start with the Bayes' rule and rearrange
\begin{eqnarray*}
p(\bs{x}_{t}|\bs{y}_{1:t})&=&\frac{p(\bs{y}_{1:t}|\bs{x}_{t})p(\bs{x}_{t})}{p(\bs{y}_{1:t})},\\
p(\bs{x}_{t}|\bs{y}_{1:t})&=&\frac{p(\bs{y}_t,\bs{y}_{1:t-1}|\bs{x}_{t})p(\bs{x}_{t})}{p(\bs{y}_t,\bs{y}_{1:t-1})},\\
p(\bs{x}_{t}|\bs{y}_{1:t})&=&\frac{p(\bs{y}_t|\bs{x}_{t},\bs{y}_{1:t-1})p(\bs{y}_{1:t-1}|\bs{x}_{t})p(\bs{x}_{t})}
{p(\bs{y}_t|\bs{y}_{1:t-1})p(\bs{y}_{1:t-1})}.
\end{eqnarray*}
We again use the Bayes' rule on $p(\bs{y}_{1:t-1}|\bs{x}_{t})$, which gives
$$
p(\bs{x}_{t}|\bs{y}_{1:t})=
\frac{p(\bs{y}_t|\bs{x}_{t},\bs{y}_{1:t-1})p(\bs{x}_{t}|\bs{y}_{1:t-1})p(\bs{y}_{1:t-1})p(\bs{x}_{t})}
{p(\bs{y}_t|\bs{y}_{1:t-1})p(\bs{y}_{1:t-1})p(\bs{x}_{t})}.\\
$$
Considering the conditional independence of $p(\bs{y}_t|\bs{x}_t,\bs{y}_{1:t-1})$,
which is expressed
by $p(\bs{y}_t|\bs{x}_t,\bs{y}_{1:t-1})=p(\bs{y}_t|\bs{x}_t)$, and cancelling out 
the $p(\bs{y}_{1:t-1})p(\bs{x}_t)$ terms we get the final formula
$$
p(\bs{x}_t|\bs{y}_{1:t})=
\frac{p(\bs{y}_t|\bs{x}_t)p(\bs{x}_t|\bs{y}_{1:t-1})}{p(\bs{y}_t|\bs{y}_{1:t-1})}.\\
$$
In the denominator, the normalizing constant is obtained by integration
$$
p(\bs{y}_t|\bs{y}_{1:t-1})=
\int p(\bs{y}_t|\bs{x}_{t})p(\bs{x}_t|\bs{y}_{1:t-1})\,d\bs{x}_{t}.
$$
As we have $p(\bs{y}_t|\bs{x}_t)=g_t(\bs{y}_t|\bs{x}_t)$, 
this finishes the proof.\hfill$\Box$

The development of the filtering density over time is split into two
sub-steps by the filtering equations. The prediction density 
$p(\bs{x}_t|\bs{y}_{1:t-1})$ is obtained in the first sub-step and, 
in the second one, it is updated  to the filtering density $p(\bs{x}_t|\bs{y}_{1:t})$
on the basis of the actual observation $\bs{y}_t$.

Speaking in the language of distributions, the filtering distribution is usually
denoted by $\pi_t$, i.e., $\pi_t(d\bs{x}_t)=p(\bs{x}_t|\bs{y}_{1:t})\,d\bs{x}_t$. 
$\pi_t$~is also alternatively referred to as the \textit{update distribution (measure)}. 
The prediction density then corresponds to the density of the so-called
\textit{prediction distribution (measure)} denoted by~$\ov{\pi}_t$, i.e.,
$\ov{\pi}_t(d\bs{x}_t)=p(\bs{x}_t|\bs{y}_{1:t-1})\,d\bs{x}_t$.

\subsection{Particle filter}
\label{SMCAlgSec}
The time evolution of the filtering distribution can be seen as a recursive 
alternation between the prediction and update distributions $\ov{\pi}_t$
and $\pi_t$. This characterization fits to the particle filter
operation because the filter alternately generates empirical prediction
and update measures.

In the particle filter, empirical measures are constructed as weighted
sums of Dirac measures localized at particles generated by the filter.
The justification of this representation stems from the Strong Law of
Large Numbers (SLLN). Assuming that $\{\bs{X}_i=\bs{x}_i\}_{i=1}^n$,
$n\in\N$ is an i.i.d. sample from a given distribution $\mu$ and 
constructing the empirical measure $\delta_n(d\bs{x})$~as
\begin{equation}
\label{PdxDef}
\delta_n(d\bs{x})=
\frac{1}{n}\sum_{i=1}^n \delta_{\bs{x}_i}(d\bs{x})=
\frac{1}{n}\sum_{i=1}^n\delta_{\bs{X}_i}(d\bs{x}),
\end{equation}
the SLLN states that for any integrable function $f$, the integral over
this empirical measure converges a.s. to the integral over the distribution~$\mu$.
Note that in (\ref{PdxDef}), the second expression points out
the random character of $\delta_n(d\bs{x})$, 
in fact, $\delta_n(d\bs{x})$ is a random measure.

Dealing with the filtering problem practically, we are not able
to directly generate i.i.d. samples from $\pi_t$ because we do not
have any closed-form representation of the filtering density at
our disposal. However, due to the product character of the joint density
$p(\bs{x}_{0:t},\bs{y}_{1:t})$, one can state an~algorithm which
recursively generates samples (particles) that are used for
constructing empirical counterparts of $\ov{\pi}_t$ and $\pi_t$
distributions.

The construction of empirical measures proceeds sequentially.
The particles generated in the previous cycle of operation
are employed in the actual cycle. A stochastic update
of particles and their weights is taken in each cycle. The
weights are updated on the basis of the actual observation.
The procedure is in fact an instance of 
the sequential Monte Carlo methods applied in the context of 
the filtering problem \cite{Doucet2001} and the algorithm follows 
the recursion described by the filtering equations.
However, there is one extension.

In the raw mode of operation, the update measure is constructed as
a non-uniformly weighted sum of Dirac measures. As explained
in \cite{Doucet2001}, as $t\in\N$ increases the distribution
of weights becomes more and more skewed and practically, 
after a few time steps, only one particle has a non-zero weight.
To avoid this degeneracy, the \textit{resampling step} is introduced.

During the resampling step, a non-uniformly weighted empirical
measure is resampled into its uniformly weighted counterpart.
The basic type of resampling is based on the idea of discarding
particles with low weights (with respect to $1/n$) and promote
those with high weights. Practically, it is done by sampling from 
the multinomial distribution $\mathcal{M}$ over original particles with 
the probabilities of selection given by particles' weights.
This type of resampling corresponds to the sampling with replacement
from the set of original particles with the probabilities of individual
selections corresponding to the individual weights. Let us stress here
that the resampled particles \textit{does not constitute an i.i.d. sample.}

We are now ready to present the operation of the particle filter in
the algorithmic way:

\begin{itemize}
\item
\textbf{0. declarations}\\
$n\in\N$ - the number of particles,\\
$T\in\N$ - the computational horizon,\\
$p_0(\bs{x}_0)$ - the initial density of $\bs{X}_0$,\\
$K_{t-1}(\bs{x}_t|\bs{x}_{t-1}),\;t=1,\dots T$ - the transition densities.
\vspace*{0.1cm}
\item
\textbf{1. initialization}\\ 
$t=0$,\\
sample $\{\ov{\bs{x}}^i_0\sim p_0(\bs{x}_0)\}_{i=1}^n$,\\ 
constitute
$\widehat{\pi}^n_0(d\bs{x}_0)=
\frac{1}{n}\sum_{i=1}^n \delta_{\ov{\bs{x}}^i_0}(d\bs{x}_0)$,\\
set ${\pi}^n_0(d\bs{x}_0)=\widehat{\pi}^n_0(d\bs{x}_0)$, i.e., 
$\{\bs{x}^i_0=\ov{\bs{x}}^i_0\}_{i=1}^n$. 
\vspace*{0.1cm}
\item
\textbf{2. sampling}\\
$t=t+1$,\\
sample $\{\ov{\bs{x}}^{i}_t\sim K_{t-1}(\bs{x}_t|\bs{x}^i_{t-1})\}_{i=1}^n$,\\
for $i=1\!:\!n$ compute
$$
\widetilde{w}(\ov{\bs{x}}^{i}_{t})=
\frac{g_t(\bs{y}_t-h_t(\ov{\bs{x}}^{i}_{t}))}
{\sum_{j=1}^n g_t(\bs{y}_t-h_t(\ov{\bs{x}}^{j}_{t}))},
$$\\
constitute
$\widehat{\pi}^n_t(d\bs{x}_{t})=
\sum_{i=1}^n \widetilde{w}(\ov{\bs{x}}^i_{t})\,
\delta_{\ov{\bs{x}}^i_{t}}(d\bs{x}_{t})$.
\item
\vspace*{0.1cm}
\textbf{3. resampling}\\
using
$\mathcal{M}(n,\widetilde{w}(\ov{\bs{x}}^1_{t}),\dots,\widetilde{w}(\ov{\bs{x}}^n_{t}))$,
resample $\{\bs{x}^i_{t}\}_{i=1}^n$ from $\{\bs{\ov{x}}^i_{t}\}_{i=1}^n$ 
and constitute\\
${\pi}^n_t(d\bs{x}_{t})=\frac{1}{n}\sum_{i=1}^n\delta_{\bs{x}^i_{t}}(d\bs{x}_{t})$.
\vspace*{0.2cm}
\item
\textbf{4.}  if $t=T$ end, else go to step 2.\\[0.2ex]
\end{itemize}
\vspace*{-0.25cm}
\begin{table}[!htb]
\centerline{\small Algorithm 1. Operation of the particle filter.}
\label{TabSMCpf}
\vspace*{-1ex}
\end{table}

The particle filter sequentially generates three empirical measures
in each single cycle of its operation. These are the empirical prediction
measure $\ov{\pi}_t^n$, the empirical update measure before resampling
$\widehat{\pi}^n_t$ and the empirical update measure after 
resampling $\pi_t^n$. The third measure then forms the empirical
counterpart of the filtering distribution $\pi_t$.

A comparison of the evolution of the empirical measures with the evolution
of the theoretical distributions can be done by means of the following schema:
\begin{figure}[!htb]
$$
\begin{array}{cccccccccc}
\pi_0
\ra&\ov{\pi}^n_1&\!\ra\!&\widehat{\pi}^n_1\ra\pi^n_1&
\ra\dots\ra&
\ov{\pi}^n_t&\!\ra\!&\widehat{\pi}^n_t\ra\pi^n_t\\
\pi_0\ra
&\ov{\pi}_{1}&\!\ra\!&\pi_{1}&
\ra\dots\ra&
\ov{\pi}_{t}&\!\ra\!&\pi_{t}
\end{array}
$$
\caption{The evolution of the empirical and theoretical distributions in the particle filter.}
\label{TmEvSMC}
\vspace*{-0.2cm}
\end{figure}

\subsection{Convergence results}
\label{ctSec}
The particle filter algorithm
is known that the empirical measures $\ov{\pi}_t^n$ and $\pi_t^n$
converge weakly a.s. (they are random measures) to their theoretical
counterparts as the number of generated particles goes to infinity.
We will not go into details of the proof of the assertion, we only mention
the result and its $L_2$ variant related to our research.

To present the convergence theorems, we denote the class of all real bounded and
continuous functions over $\R^{d_x}$ by $\mathcal{C}_b(\mathbb{R}^{d_x})$,
the supremum norm of a function $f:\R^{d_x}\ra\R$ by 
$||f||_{\infty}$, i.e., $||f||_{\infty}=\sup_{\bs{x}}\{|f(\bs{x})|\}$,
and the integral of $f$ over the measure $\mu$ by $\mu f$.
Further, it is assumed that the transition
kernels of the signal process possess the Feller property. That is, 
$K_{t-1}f\in\mathcal{C}_b(\mathbb{R}^{d_x})$ for any
$f\in\mathcal{C}_b(\mathbb{R}^{d_x})$ and $t\in\N$, where
$(K_{t-1}f)(\bs{x}_{t-1})=\int f(\bs{x}_t) K_{t-1}(d\bs{x}_t|\bs{x}_{t-1})$.
The other assumption is that the densities $g_t(\bs{y}_t|\,\cdot\,)$ of (\ref{gtdef}),
\mbox{$t\in\N$} are bounded, continuous and strictly positive functions.

\begin{theorem}
\label{asctTheorem}
Let $\{\ov{\pi}_t^n\}_{t=1}^{T}$ and $\{\pi_t^n\}_{t=1}^{T}$ 
be the sequences of empirical measures generated by the particle filter 
for some fixed observation history $\{\bs{Y}_t=\bs{y}_t\}_{t=1}^{T}$, $T\in\N$.
Then for all $t\in \{1,\dots,T\}$ and \mbox{$f\in\mathcal{C}_b(\mathbb{R}^{d_x})$},
$$
\lim_{n\ra\infty} |\ov{\pi}_t^nf-\ov{\pi}_tf|= 0\;\;a.s.,\;\;
\lim_{n\ra\infty} |\pi_t^nf-\pi_tf|= 0\;\;a.s.
$$
\end{theorem}

\begin{proof}
See \cite{Doucet2001}, Chapter 2 for a discussion of the convergence theorems.
Other source is \cite{Crisan2002}, Section IV. Paper \cite{Crisan2013} 
has a proof even for unbounded functions in Proposition 1(b).
\end{proof}

In our research we employ the $L_2$ version of the theorem
for $\pi^n_t$. It reads as follows:

\begin{theorem}
\label{ctTheorem}
Let $\{\pi_t^n\}_{t=1}^{T}$ 
be the sequence of empirical measures generated by the particle filter
for some fixed observation history $\{\bs{Y}_t=\bs{y}_t\}_{t=1}^{T}$, $T\in\N$.
Then for all $t\in\{1,\dots,T\}$ and 
\mbox{$f\in\mathcal{C}_b(\mathbb{R}^{d_x})$},
\begin{equation}
\label{ctformula}
\E[|\pi_t^nf-\pi_tf|^2]\leq \frac{c^2_t||f||^2_{\infty}}{n}
\end{equation}
with $c_t>0$ being a constant for fixed $t\in\{1,\dots,T\}$.
\end{theorem}

\begin{proof}
In this formulation, the theorem is presented 
in~\cite{Crisan2002}, Section~V
(authors use $c_t$ instead ours $c^2_t$).
\end{proof}

Remark that the $L_1$ version, i.e.,
$\E[|\pi_t^nf-\pi_tf|]$, is treated in~\cite{Doucet2001}, 
Theorem 2.4.1. The  theorem is mentioned for general $L_p$
norm, $p\geq 1$ in \cite{Crisan2013}, Proposition 1(a).

The theorem holds also for the class
$\mathcal{C}^{\C}_b(\mathbb{R}^{d_x})$ of bounded and continuous
complex functions of real variables over $\R^{d_x}$. That is, it holds 
also for functions $h:\R^{d_x}\ra\C$, $h(\bs{x})=f(\bs{x})+\mathrm{i}g(\bs{x})$,
$f,g\in\mathcal{C}_b(\R^{d_x})$, where $\mathrm{i}$ denotes the imaginary unit.
Clearly, the extension on complex functions is due to the triangle inequality 
for the absolute value (the modulus) of a complex number.

\section{Kernel methods}
\label{SecIII}
Kernel methods are widely used for nonparametric estimation of densities
of probability distributions with the vast literature available on the topic.
Here we review the very basics of the related methodology.
We focus in more details on the application of Fourier analysis in this field.
Our review is mainly based on the standard works of
\cite{Silverman1986} and \cite{Wand1995}, 
and the recent book by Tsybakov \cite{Tsybakov2009}.

\subsection{Basics of kernel methods}
Let $\bs{X}_1,\dots, \bs{X}_n$, $n\in\N$ 
be a set of independent random variables identically distributed as the 
real random variable~$\bs{X}:(\Omega,\A)\ra(\R^d,\B(\R^d))$.
Let the distribution of $\bs{X}$ have the density \mbox{$f:\R^d\ra[0,\infty)$}
with respect to the $d$-dimensional Lebesgue measure.
A nonparametric kernel density estimate of $f$ is constructed
on the basis of an i.i.d. sample \mbox{$\{\bs{X}_i=\bs{x}_i\}_{i=1}^n$} 
from the distribution of $\bs{X}$. The estimate is constructed as
a generalization of the classical histogram by replacing the indicator function,
which specifies individual bins of the histogram, by a more general 
function \mbox{$K:\R^d\ra\R$} which is commonly referred to 
as the \textit{kernel function} or simply as the~\textit{kernel}.

The definition formula of the standard $d$-variate nonparametric 
kernel density estimate writes as
\begin{equation}
\label{fnDdim}
\hat{f}_n(\bs{x})=
\frac{1}{nh^d}
\sum_{i=1}^n 
K\!\!\left(\frac{\bs{x}-\bs{x}_i}{h}\right)=
\frac{1}{nh^d}
\sum_{i=1}^n 
K\!\!\left(\frac{\bs{x}-\bs{X}_i}{h}\right).
\end{equation}
In the formula, the second expression points out the random
character of the estimate. That is, for each  $\bs{x}\in\R^d$, 
the estimate $\hat{f}_n(\bs{x})$ constitutes a random variable 
whose distribution is determined by the distribution of $\bs{X}$
and by the value of the parameter $h>0$ which is called 
the~\textit{bandwidth}.

Due to the random character of $\hat{f}_n(\bs{x})$, there is the relevant
question of the consistency and unbiasedness of the estimate.
In the univariate case, the classical result of Parzen \cite{Parzen1962}
(see also \cite{Silverman1986}, p.~71) states the conditions under
which the estimate is consistent. The result extends
on the multivariate case, see e.g. \cite{Givens1995}. The conditions
are imposed on the properties of the kernel function and on the evolution
of the bandwidth $h$ in dependence on the sample size $n\in\N$. We
mention only that $h$ is required to evolve in such a way that 
1)~$\lim_{n\ra\infty} h(n)=0$ and 2) $\lim_{n\ra\infty} nh^d(n)=\infty$.

The investigation on the bias of $\hat{f}_n(\bs{x})$ is closely related
to the investigation on the quality of the estimate 
in terms of the \textit{mean squared error} - 
$\mathrm{MSE}_{\bs{x}}(\hat{f}_n)$.
For a fixed point $\bs{x}\in\R^d$, the error is specified as
$\mathrm{MSE}_{\bs{x}}(\hat{f}_n)=\E[(\hat{f}_n(\bs{x})-f(\bs{x}))^2]$.
Employing properties of mean and variance, 
it writes as
\begin{equation}
\label{MSEdef}
\mathrm{MSE}_{\bs{x}}(\hat{f}_n)
=(\E[\hat{f}_n(\bs{x})]-f(\bs{x}))^2+var[\hat{f}_n(\bs{x})]
=(b[\hat{f}_n](\bs{x}))^2+\sigma^2[\hat{f}_n](\bs{x}),
\end{equation}
where the term $b[\hat{f}_n](\bs{x})=\E[\hat{f}_n(\bs{x})]-f(\bs{x})$ 
is the \textit{bias} and $\sigma^2[\hat{f}_n](\bs{x})=var[\hat{f}_n(\bs{x})]$
the \textit{variance} of the kernel density estimate $\hat{f}_n(\bs{x})$
at the point $\bs{x}\in\R^d$.

The $\mathrm{MSE}_{\bs{x}}(\hat{f}_n)$ is the local measure of the quality
of the estimate. It is desirable to have also a corresponding global measure.
Expectedly, such the measure deals with local errors accumulated over the 
whole domain of the estimated density. Mathematically, the accumulation is 
performed by integration. This leads to the notion of the 
\textit{mean integrated squared error} (MISE) of a kernel density estimate.

The MISE of the kernel density estimate $\hat{f}_n$ is defined and expressed 
on the basis of (\ref{MSEdef}) using the Fubini's theorem as
\begin{eqnarray}
\label{1DMISE}
\mathrm{MISE}(\hat{f}_n)&\!\!=
\!\!&\E\int[(\hat{f}_n(\bs{x})-f(\bs{x}))^2]\,d\bs{x}
=\int \MSE_{\bs{x}}(\hat{f}_n)\,\bs{x}\nonumber\\
&\!\!=\!\!&\int(\E[\hat{f}_n(\bs{x})]-f(\bs{x}))^2\,d\bs{x}+
\int var [\hat{f}_n(\bs{x})]\,d\bs{x}\nonumber\\
&\!\!=\!\!&\int(b[\hat{f}_n](\bs{x}))^2\,d\bs{x}+
\int \sigma^2[\hat{f}_n](\bs{x})\,d\bs{x}.\nonumber
\label{MISEhexact}
\end{eqnarray}
The formula consists of two summands which are the integrated
versions of the squared bias and variance terms of the
$\mathrm{MSE}_{\bs{x}}(\hat{f}_n)$.
The value of the MISE($\hat{f}_n$) depends on the value of the bandwidth $h$.

It is a standard observation that the bias and variance terms 
behave in the opposite way with respect to the magnitude of the bandwidth. 
That is, for $n\in\N$ fixed, if $h$ decreases, i.e., if $h\ra 0$,
then the bias goes to zero, and we have the asymptotic unbiasedness
of the $\hat{f}_n(\bs{x})$ estimate. However, the variance increases. 
If $h$ increases, i.e., if $h\ra\infty$, the bias increases too,
but the variance term diminishes.  Thus, we encounter here the situation
of the \textit{bias-variance trade-off} when minimizing the 
$\mathrm{MISE}(\hat{f}_n)$ by adjusting the bandwidth $h$.

The specification of the optimal value $h^*_\mathrm{MISE}$ 
minimizing (\ref{1DMISE}) can be made analytically only if (\ref{1DMISE}) 
has a closed-form expression. This is known only in some specific cases, 
for example, when the estimated density $f$
is a convex sum of normal densities, see \cite{Silverman1986}, p.~37 or
\cite{Wand1995}, p.~102 for the related explicit formulas for 
$\mathrm{MISE}(\hat{f}_n)$.
To deal with the minimization problem generally, the widely used approach 
is to investigate the asymptotic behavior of the MISE with respect
to the sample size $n\in\N$ going to infinity (AMISE analysis). 
The result based on the Taylor's expansion of the estimated density $f$ states 
(\cite{Silverman1986}, p. 85, \cite{Wand1995}, p.~99) that
\begin{equation}
\label{AMISE}
\mathrm{MISE}(\hat{f}_n)\approx
n^{-1}h^{-d}R(K)+\frac{1}{4}h^4(\mu_2(K))^2
\int (\nabla^2\!f(\bs{x}))^2\,d\bs{x}
\end{equation}
for $R(K)=\int K^2(\bs{u})\,d\bs{u}$, $\mu_2(K)=\int u_1^2K(\bs{u})\,d\bs{u}$,
$\nabla^2\!f(\bs{x})=\sum_{i=1}^d (\partial^2/\partial x^2_i)f(\bs{x})$.
Using standard calculus, the minimizer
of the above formula reads as
\begin{equation}
\label{nDAMISE}
h_{\mathrm{AMISE}}^*=
\left[\frac{d\cdot R(K)}{\mu_2(K)^2
\int (\nabla^2\!f(\bs{x}))^2\,d\bs{x}}
\;\frac{1}{n}\right]^{1/(d+4)}
\!\!\!\!\!\!\!\!\!\!\!\!\!\!\!\!\!\!.
\end{equation}

In (\ref{AMISE}), the terms $R(K)$ and $\mu_2(K)$ can be further
minimized over a set of appropriate kernels. The minimizer is known
as the \textit{Epanechnikov kernel} which is specified as
$K_{e}(\bs{u})=\frac{1}{2}\vartheta^{-1}_d(d+2)(1-||\bs{u}||^2)_+$
where $\vartheta_d$ is the volume of the $d$-dimensional unit sphere,
$||\cdot||$ is the Euclidean norm and $(\cdot)_+=\max\{0,\cdot\}$ 
is the positive part.

AMISE analysis represents the standard approach to the analytic 
specification of a suitable value of the bandwidth when constructing a kernel
density estimate, even though the specification of $h^*_{\mathrm{AMISE}}$
requires the knowledge of partial derivatives of the density $f$ under
estimation. Typically, to overcome the deadlock, the respective
entities are somehow estimated from data \cite{Wand1995}.

However, in Section 1.2.4 of his book \cite{Tsybakov2009}, Tsybakov 
provides a deeper criticism of the asymptotic approach. 
It stems from the fact that the optimality of $h^*_{\mathrm{AMISE}}$
is related to a~\textit{fixed density~$f$} and not to a well defined class
of densities. In Proposition 1.7, Tsybakov shows 
that for a given fixed density $f$ it is possible to construct such a non-negative
kernel estimate that the MISE($\hat{f}_n$) diminishes, but this cannot 
be done uniformly over a sufficiently broad class of densities. Examples of such classes, 
e.g. H\"{o}lder, Sobolev or Nikol'ski classes, are presented in \cite{Tsybakov2009}.
The Sobolev class is treated in Definition~\ref{defSob} below.

Based on this criticism, Tsybakov presents a different approach 
to the MISE analysis in Section~1.3 of \cite{Tsybakov2009}. 
The approach relies on Fourier analysis.

\subsection{Fourier analysis}
In this section, we deal with the application of Fourier ana\-ly\-sis in
the area of nonparametric kernel density estimation. We mainly follow
the presentation of Tsybakov given in Chapter~1 of \cite{Tsybakov2009}.
In \cite{Tsybakov2009}, results are provided for the univariate
case. In order to the results could be applied in our research
presented in Section~\ref{SecIV}, we have extended them into multiple dimensions.

In the probability theory, Fourier analysis is intimately interconnected with
the notion of the characteristic function. 
Let $\bs{X}:(\Omega,\A)\ra(\R^d,\B(\R^d))$ be a $d$-variate
real random vector with the joint distribution $\mu(d\bs{x})$. 
The characteristic function $\phi(\bs{\omega}):\R^d\ra\C$ of $\bs{X}$ 
is defined as the integral transform
\begin{equation}
\label{charfce}
\phi(\bs{\omega})=\E[e^{\mathrm{i}\lss\bs{\omega},\bs{X}\rss}]=
\int e^{\mathrm{i}\lss\bs{\omega},\bs{x}\rss}\,
\mu(d\bs{x}),\;\;\bs{\omega}\in\R^d,
\end{equation}
where $\lss\cdot,\!\cdot\rss$ denotes the dot product.
It is well known that the transform provides the complete characterization
of the distribution of $\bs{X}$; and we often speak
about the Fourier transform of the random vector $\bs{X}$. 

The other quite common view of the Fourier transform comes from the
area of applied mathematics. Let  $f:\R^d\ra\R$ be an integrable
function (a signal in electrical engineering), i.e., let $f\in L_1(\R^d)$, 
then its Fourier transform is specified as
\begin{equation}
\label{Ftransform}
\F[f](\bs{\omega})=\int e^{\mathrm{i}\lss\bs{\omega},\bs{x}\rss}
f(\bs{x})\,d\bs{x},
\;\;\bs{\omega}\in\R^d.
\end{equation}
Formula (\ref{Ftransform}) can be treated as the special case of formula
(\ref{charfce}) when the distribution of $\bs{X}$ is absolutely continuous
with respect to the $d$-dimensional Lebesgue measure and
has the density~$f$, i.e., $\mu(d\bs{x})=f(\bs{x})\,d\bs{x}$. On the
other hand, in (\ref{Ftransform}), $f$ need not be necessarily a density,
only the integrability is assumed.

Let $f,g\in L_1(\R^d)\cap L_2(\R^d)$, i.e., we consider functions both
$L_1$ and $L_2$ integrable over $\R^d$, then the following properties
of the multivariate Fourier transform are relevant to our research:
\begin{itemize}
\item continuity: $\F[f]$ is uniformly continuous on $\R^d$,
\item linearity:
$\F[af+bg](\bs{\omega})=a\F[f](\bs{\omega})+b\F[g](\bs{\omega}),\;\;\;a,b\in\R$,
\item shifting:
$\F[f(\bs{x}-\bs{s})](\bs{\omega})=
e^{\mathrm{i}\lss \bs{\omega},\bs{s}\rss}\F[f](\bs{\omega}),\;\;\bs{s}\in\R^d$,
\item scaling:\;
$\F[f(\bs{x}/h)/h^d](\bs{\omega})=\F[f](h\bs{\omega}),\;\;h>0$,
\item shifting \& scaling:
$\F[f((\bs{x}-\bs{s})/h)/h^d]=
e^{\mathrm{i}\lss \bs{\omega},\bs{s}\rss}\F[f](h\bs{\omega}),\;\;\bs{s}\in\R^d$,
\item complex conjugate:
$\ov{\F[f](\bs{\omega})}=\F[f](-\bs{\omega})$,
\item convolution:
$\F[f*g](\bs{\omega})=\F[f](\bs{\omega})\F[g](\bs{\omega})$,
\item derivative:
$\F[f^{(m)}_{i_1,\dots,i_d}](\bs{\omega})\!=\!
(\mathrm{-i})^m
(\omega^{i_1}_1\cdot\dots\cdot\omega^{i_d}_d)\F[f](\bs{\omega})$,
\item symmetry:
if $f(-\bs{x})=\!f(\bs{x})$, then $\F[f](-\bs{\omega})\!=\!\F[f](\bs{\omega})$,
\item isometry, due to the Plancheler's formula for $f\in L_2(\R^d)$:
$$
\int f^2(\bs{x})\, d\bs{x}=
\frac{1}{(2\pi)^d}\int |\F[f](\bs{\omega})|^2\,d\bs{\omega}.
$$
\end{itemize}

Now, the uniformly weighted sum of Dirac measures
$\delta_n(d\bs{x})$ introduced in formula (\ref{PdxDef})
represents the probability distribution which does not have 
any density with respect to the corresponding Lebesgue measure. 
Its characteristic function $\phi_n(\bs{\omega})$ is specified as
\begin{equation}
\label{phindef}
\phi_n(\bs{\omega})=\int e^{\mathrm{i}\lss \bs{\omega},\bs{x} \rss} 
\delta_n(d\bs{x})=
\frac{1}{n}\sum_{j=1}^n e^{\mathrm{i}\lss \bs{\omega},\bs{X}_j\rss},
\;\;\bs{\omega}\in\R^d.
\end{equation}
Note that $\phi_n(\bs{\omega})$ constitutes a random
variable for $\bs{\omega}\in\R^d$ being fixed.

Under the assumption of $L_1(\R^d)$ integrability of the employed
kernel $K$, we can consider the Fourier transform of the 
multivariate density kernel estimate~(\ref{fnDdim}).
Using the linearity and the shifting~\&~scaling property of
the Fourier transform, $\F[\hat{f}_n](\bs{\omega})$ is specified by formula
\begin{equation}
\F[\hat{f}_n](\bs{\omega})=
\frac{1}{n}\sum_{j=1}^n
\F\left[\frac{1}{h^d}K\!\!\left(\frac{\bs{x}-\bs{X}_j}{h}\right)\right]
=
\frac{1}{n}\sum_{j=1}^n 
e^{\mathrm{i}\lss\bs{\omega},\bs{X}_j\rss}\F[K](h\bs{\omega}).
\end{equation}

Writing $K_{\F}(\bs{\omega})$ for $\F[K](\bs{\omega})$ we obtain
the compact expression of $\hat{f}_n$ in the form
\begin{equation}
\label{Fconv}
\F[\hat{f}_n](\bs{\omega})=\phi_n(\bs{\omega})K_{\F}(h\bs{\omega}).
\end{equation}
This shows that the standard kernel estimator which is based on an i.i.d.
sample is obtained by the convolution of the employed kernel with
the uniformly weighted sum of Dirac measures corresponding to the
sample.

To proceed with the investigation of the MISE of density kernel estimates
in the frequency domain, we present a multivariate version of
Lemma 1.2 from~\cite{Tsybakov2009}.
\begin{lemma}
\label{LmTheorem}
Let $\{\bs{X}_j\}_{j=1}^n$ be an i.i.d. sample from a distribution with the density~$f$.
Let the characteristic function of $\bs{X}_j$ be $\phi(\bs{\omega})$.
Then for $\phi_n$ of $\,\mathrm{(\ref{phindef})}$ we have
\begin{flalign*}
\begin{array}{ll}
\mathrm{(i)}&\E[\phi_n(\bs{\omega})]=\phi(\bs{\omega}),\\
\mathrm{(ii)}&\E[|\phi_n(\bs{\omega})|^2]=\left(1-\frac{1}{n}\right)|\phi(\bs{\omega})|^2
+\frac{1}{n},\\
\mathrm{(iii)}&\E[|\phi_n(\bs{\omega})-\phi(\bs{\omega})|^2]=
\frac{1}{n}(1-|\phi(\bs{\omega})|^2).
\end{array} &&
\end{flalign*}
\end{lemma}

\noindent
\textbf{Proof.} To show (i), consider the i.i.d. character of $\{\bs{X}_j\}_{j=1}^n$,
$$
\E[\phi_n(\bs{\omega})]=\frac{1}{n}\sum_{j=1}^n
\int e^{\mathrm{i}\lss\bs{\omega},\bs{x}\rss}f(\bs{x})\,d\bs{x}
=\frac{1}{n}\sum_{j=1}^n\phi(\bs{\omega})
=\phi(\bs{\omega}).
$$
To show (ii), note that
\begin{eqnarray}
\E [|\phi_n(\bs{\omega})|^2]&=&
\E [\phi_n(\bs{\omega})\ov{\phi_n(\bs{\omega})}]=
\E [\phi_n(\bs{\omega})\phi_n(-\bs{\omega})]\nonumber\\
&=&\E\left[\frac{1}{n^2}\sum_{j,k:j\not= k}
e^{\mathrm{i}\lss\bs{\omega},\bs{X}_j\rss}e^{-\mathrm{i}\lss\bs{\omega},\bs{X}_k\rss}\right]
+\frac{n}{n^2}\nonumber\\
&=&\frac{1}{n^2}\sum_{j,k:j\not= k}
\E\left[e^{\mathrm{i}\lss\bs{\omega},\bs{X}_j\rss}\right]
\E\left[e^{-\mathrm{i}\lss\bs{\omega},\bs{X}_k\rss}\right]
+\frac{1}{n}\nonumber\\
&=&\frac{n^2-n}{n^2}\phi(\bs{\omega})\phi(-\bs{\omega})+\frac{1}{n}\nonumber\\
&=&\left(1-\frac{1}{n}\right)|\phi(\bs{\omega})|^2+\frac{1}{n}.
\end{eqnarray}
Case (iii) folows from (ii) a (i). Indeed,
$\E[|\phi_n(\bs{\omega})-\phi(\bs{\omega})|^2]=$
\begin{eqnarray}
&=&\E[
(\phi_n(\bs{\omega})-\phi(\bs{\omega}))\ov{(\phi_n(\bs{\omega})-\phi(\bs{\omega}))}]\nonumber\\
&=&\E[
(\phi_n(\bs{\omega})-\phi(\bs{\omega}))(\phi_n(-\bs{\omega})-\phi(-\bs{\omega}))]\nonumber\\
&=&\E[
|\phi_n(\bs{\omega})|^2-\phi_n(\bs{\omega})\phi(-\bs{\omega})
-\phi(\bs{\omega})\phi_n(-\bs{\omega})+|\phi(\bs{\omega})|^2]\nonumber\\
&=&\E[
|\phi_n(\bs{\omega})|^2]-\phi(\bs{\omega})\phi(-\bs{\omega})
-\phi(\bs{\omega})\phi(-\bs{\omega})+|\phi(\bs{\omega})|^2\nonumber\\
&=&\E[
|\phi_n(\bs{\omega})|^2]-2|\phi(\bs{\omega})|^2+|\phi(\bs{\omega})|^2\nonumber\\
&=&\left(1-\frac{1}{n}\right)|\phi(\bs{\omega})|^2+\frac{1}{n}-|\phi(\bs{\omega})|^2\nonumber\\
&=&\frac{1}{n}(1-|\phi(\bs{\omega})|^2)\nonumber.
\end{eqnarray}
This concludes the proof.\hfill$\Box$\\

Let us assume that both density $f$ and kernel $K$ belong also to $L_2(\R^d)$.
Then employing the Plancherel's theorem and~(\ref{Fconv}), we get for
the MISE of (\ref{1DMISE}) the expression
\begin{equation}
\label{LLP}
\mathrm{MISE}(\hat{f}_n)=\frac{1}{(2\pi)^d}\;
\E\!\int |\phi_n(\bs{\omega})K_{\F}(h\bs{\omega})-\phi(\bs{\omega})|^2\,d\bs{\omega}.
\end{equation}

The next theorem provides the exact computation of the MISE($\hat{f}_n$)
for any fixed $n\in\N$.
\begin{theorem}
\textit{Let $f\in L_2(\R^d)$ be a density and
$K\in L_1(\R^d)\cap L_2(\R^d)$ a kernel.
Then for all $n\geq 1$ and $h>0$ the MISE of the 
i.i.d. based kernel estimator $\hat{f}_n$ of 
$\mathrm{(\ref{fnDdim})}$ has the form}
\begin{eqnarray}
\label{MISEFourier}
\mathrm{MISE}(\hat{f}_n)&=&\;\;\;\frac{1}{(2\pi)^d}
\left[\int|1-K_{\F}(h\bs{\omega})|^2|\phi(\bs{\omega})|^2\,d\bs{\omega}+
\frac{1}{n}\int |K_{\F}(h\bs{\omega})|^2\,d\bs{\omega}\right] \nonumber\\
&&-\frac{1}{(2\pi)^d}\frac{1}{n}
\int |\phi(\bs{\omega})|^2 |K_{\F}(h\bs{\omega})|^2\,d\bs{\omega}.
\end{eqnarray}
\end{theorem}

\noindent
\textbf{Proof.} As $\phi, K\in L_2(\R^d)$ and $|\phi(\bs{\omega})|\leq 1$
for all $\bs{\omega}\in\R^d$, all the integrals are finite.
To obtain the Fourier MISE formula it suffices to develop (\ref{LLP}),
\begin{eqnarray*}
&&\E\int|\phi_n(\bs{\omega})K_{\F}(h\bs{\omega})-\phi(\bs{\omega})|^2\,d\bs{\omega}\\
&=&\E\int|(\phi_n(\bs{\omega})-\phi(\bs{\omega}))K_{\F}(h\bs{\omega})-
(1-K_{\F}(h\bs{\omega}))\phi(\bs{\omega})|^2\,d\bs{\omega}\\
&=&\E\int((\phi_n(\bs{\omega})-\phi(\bs{\omega}))K_{\F}(h\bs{\omega})-
(1-K_{\F}(h\bs{\omega}))\phi(\bs{\omega}))\\
&&\hspace{0.5cm}\cdot\;(\ov{(\phi_n(\bs{\omega})-\phi(\bs{\omega}))K_{\F}(h\bs{\omega})}-
\ov{(1-K_{\F}(h\bs{\omega}))\phi(\bs{\omega})})\,d\bs{\omega}\\
&=&\int\E [|\phi_n(\bs{\omega})-\phi(\bs{\omega})|^2]|K_{\F}(h\bs{\omega})|^2+
|1-K_{\F}(h\bs{\omega})|^2|\phi(\bs{\omega})|^2\\
&&\hspace{0.4cm}+\;(\E [(\phi_n(\bs{\omega})]-\phi(\bs{\omega}))K_{\F}(h\bs{\omega})
\,\ov{(1-K_{\F}(h\bs{\omega}))\phi(\bs{\omega}))}\\[0.2cm]
&&\hspace{0.4cm}+\;(1-K_{\F}(h\bs{\omega}))\phi(\bs{\omega}))\,
(\E [\ov{\phi_n(\bs{\omega})}]-\ov{\phi(\bs{\omega}}))\ov{K_{\F}(h\bs{\omega})}
\,d\bs{\omega}\\[0.2cm]
&=&\int\E [|\phi_n(\bs{\omega})-\phi(\bs{\omega})|^2]|K_{\F}(h\bs{\omega})|^2+
|1-K_{\F}(h\bs{\omega})|^2|\phi(\bs{\omega})|^2\,d\bs{\omega}\\
&=&\frac{1}{n}\int(1-|\phi(\bs{\omega})|^2)|K_{\F}(h\bs{\omega})|^2\,d\bs{\omega}
+\int |1-K_{\F}(h\bs{\omega})|^2|\phi(\bs{\omega})|^2\,d\bs{\omega}
\end{eqnarray*}
After rearranging we obtain the assertion of the theorem.\hfill$\Box$

We are now going to discuss the individual terms in
the Fourier MISE formula (\ref{MISEFourier}).
We start with the notion of the \textit{order of a kernel}.

\begin{definition}
\label{defell}
Let $\ell\geq 1$ be an integer. We say that the kernel $K\!\!:\R^d\ra\R$ is of
order $\ell$, if $K$ is $L_1(\R^d)\cap L_2(\R^d)$ integrable, 
its Fourier transform $K_\F(\bs{\omega})$ is real, satisfies \mbox{$K_\F(\bs{0})=1$}
and has all partial derivatives 
$K^{(m)}_{\F,i_1,\dots i_d}=\partial^mK_\F/\partial_{i_1}\dots\partial_{i_d}$,
$m=i_1+\dots+i_d$, $m\in\N$ up to the $\ell$-th order and it holds that
$K^{(m)}_{\F,i_1,\dots i_d}(\bs{0})=0$ for all $m=1,\dots,\ell$.
\end{definition}

Remark that the above definition imposes the following conditions 
on a multivariate kernel to be of order $\ell\geq 1$, $\ell\in\N$:
\begin{itemize}
\item
$\int K(\bs{u})\,d\bs{u}=1$,
\item
$\int u^{i_1}_{1}\cdots u^{i_d}_{d} K(\bs{u})\,d\bs{u}=0$ 
for $m=1,\dots,\ell$.
\end{itemize}
Indeed, at the origin we have
$K_{\F}(\bs{0})=\int e^{\mathrm{i}\lss\bs{0},\bs{u}\rss}K(\bs{u})\,d\bs{u}
=\int K(\bs{u})\,d\bs{u}=1$.
For the \mbox{$m$-th} partial derivative, we get  
$$
K^{(m)}_{\F,i_1,\dots i_d}(\bs{\omega})=
\int (\mathrm{i}u_1)^{i_1}\cdots(\mathrm{i}u_d)^{i_d}\,
e^{\mathrm{i}\lss\bs{\omega},\bs{u}\rss}K(\bs{u})\,d\bs{u},
$$ 
hence $0=K^{(m)}_{\F,i_1,\dots i_d}(\bs{0})=\mathrm{i}^m
\int u^{i_1}_{1}\cdots u^{i_d}_{d} K(\bs{u})\,d\bs{u}$.

From the remark, it follows that kernels
of order $\ell\geq 2$ must take negative values. If such kernels
are allowed in kernel estimates, then $\hat{f}_n$ of (\ref{fnDdim})
may also take negative values. However, this is not a serious drawback
because we can always take as the final estimate the positive part of $\hat{f}_n$,
i.e., $\hat{f}^+_n=\max\{0,\hat{f}_n\}$. At each point $\bs{x}\in\R^d$, the
$\mathrm{MSE}_{\bs{x}}$ of $\hat{f}^+_n(\bs{x})$ is always smaller than that
of negative $\hat{f}_n(\bs{x})$. Therefore we have also
$\mathrm{MISE}(\hat{f}^+_n)\leq \mathrm{MISE}(\hat{f}_n)$.\\

\subsubsection{The first term}
For the first term in the Fourier MISE formula
(\ref{MISEFourier}), we are able to say something more
specific if we consider the order of the kernel involved
in the estimate.

\begin{theorem}
\label{bdTheorem}
Let $K\!\!:\R^d\ra\R$ be a kernel of order $\ell\geq 1$,
$\ell\in\N$. Then there exists a constant $A>0$ such that
\begin{equation}
\label{Acondition}
\mathrm{sup}_{{\bs{\omega}\in\R^d\backslash\{\bs{0}\}}}
\frac{|1-K_{\F}(\bs{\omega})|}{||\bs{\omega}||^{\ell}}
\leq A,
\end{equation}
and
\begin{equation}
\label{1stTerm}
\int |1-K_{\F}(h\bs{\omega})|^2|\phi(\bs{\omega})|^2 d\bs{\omega}
\leq
A^2h^{2\ell}\int ||\bs{\omega}||^{2\ell}|\phi(\bs{\omega})|^2 d\bs{\omega}
\end{equation}
for any function $f$ with the Fourier transform $\phi(\bs{\omega})$ and $h>0$.
\end{theorem}

\noindent
\textbf{Proof.}
We employ the multidimensional Taylor's theorem. Because the kernel $K$ 
is of order $\ell\geq 1$, its Fourier transform $K_\F(\bs{\omega})$ is real
and by the Taylor's theorem
\begin{equation*}
K_{\F}(\bs{\omega})
=K_{\F}(\bs{0})+
\frac{1}{1!}\sum_{i=1}^d K^{(1)}_{\F,i}(\bs{0})\,\omega_i
+\dots+
\frac{1}{\ell!}\!\!\!\!\!
\sum_{i_1,\dots,i_d=1}^d\!\!\!\!\!\!
K^{(\ell)}_{\F,i_1,\dots,i_d}({\bs{0}})\,
\omega_{i_1}\cdots\omega_{i_d}+R_\ell(\bs{\omega})
\end{equation*}
with 
$
\lim_{\bs{\omega}\ra\bs{0}}
R_{\ell}(\bs{\omega})/||\bs{\omega}||^{\ell}=0
$
for the reminder, where $||\cdot||$ is the Euclidean norm.

Because the involved partial derivatives equal to zero, the remainder writes
$R_{\ell}(\bs{\omega})=K_{\F}(\bs{\omega})-K_{\F}(\bs{0})=K_{\F}(\bs{\omega})-1$
and $\lim_{\bs{\omega}\ra\bs{0}} |1-K_{\F}(\bs{\omega})|/||\bs{\omega}||^{\ell}=0$
by the Taylor's theorem.

Let us define $A_{\ell}(\bs{\omega})=|1-K_{\F}(\bs{\omega})|/||\bs{\omega}||^{\ell}$
for $\bs{\omega}\not=\bs{0}$, and $A_{\ell}(\bs{0})=0$.
The function \mbox{$A_{\ell}:\R^d\ra [0,\infty)$} is continuous on $\R^d$
and attains its maximum on the 
unit ball \mbox{$||\bs{\omega}||\leq 1$}. We denote this maximum by $M_1$,
$M_1=\max_{\{||\bs{\omega}||\,\leq\,1\}}\{A_\ell(\bs{\omega})\}$.
Because $K\in L_1(\R^d)$, we have $0\leq |K_{\F}(\bs{\omega})|\leq M_2<\infty$.
Indeed, $|K_{\F}(\bs{\omega})|\leq \int |e^{\mathrm{i}\lss\bs{\omega},\bs{u}\rss}|\,
|K(\bs{u})|\,d\bs{u}\leq \int |K(\bs{u})|\,d\bs{u}=M_2<\infty$. Therefore, 
$|1-K_{\F}|/||\bs{\omega}||^{\ell}\leq 1+M_2$ for $||\bs{\omega}||>1$.
Composing both cases one gets
$A_{\ell}(\bs{\omega})\leq \max\{M_1,1+M_2\}=A<\infty$
for $\bs{\omega}\in\R^d$.

The inequality (\ref{1stTerm}) is implied by (\ref{Acondition}) as follows:
\begin{eqnarray*}
\mathrm{sup}_{\bs{\omega}\in\R^d\backslash\{\bs{0}\}}\;
\frac{|1-K_{\F}(h\bs{\omega})|}{||h\bs{\omega}||^{\ell}}
\!\!\!&\leq&\!\!\! A,\\
|1-K_{\F}(h\bs{\omega})|&\leq& A||h\bs{\omega}||^{\ell},\\
|1-K_{\F}(h\bs{\omega})|^2&\leq& A^2||h\bs{\omega}||^{2\ell},\\
|1-K_{\F}(h\bs{\omega})|^2|\phi(\bs{\omega})|^2
\!\!\!&\leq&\!\!\!
A^2h^{2\ell}||\bs{\omega}||^{2\ell}|\phi(\bs{\omega})|^2,\\
\int |1-K_{\F}(h\bs{\omega})|^2|\phi(\bs{\omega})|^2 d\bs{\omega}
\!\!\!&\leq&\!\!\!
A^2h^{2\ell}\!\!\int ||\bs{\omega}||^{2\ell}|\phi(\bs{\omega})|^2 d\bs{\omega}.
\end{eqnarray*}
This concludes the proof.\hfill$\Box$\\

The other terms in formula (\ref{MISEFourier}) refer to individual
properties of the kernel and density under considerations. We mention
only two straightforward observations.\\

\subsubsection{The second term}
The second term can be directly translated from the frequency to the
``time" domain by the Plancherel's theorem and the scaling property of 
the Fourier transform:
\begin{equation}
\label{2ndMISEterm}
\frac{1}{n}
\int |K_{\F}(h\bs{\omega})|^2\,d\bs{\omega}=
\frac{(2\pi)^d}{nh^{2d}}\int K^2(\bs{x}/h)\,d\bs{x}
=
\frac{(2\pi)^d}{nh^d}\int K^2(\bs{u})\,d\bs{u}.
\end{equation}

\subsubsection{The third term}
The third term is actually the correction term.
For this term we have the following inequality:
\begin{eqnarray*}
\frac{1}{(2\pi)^d}\frac{1}{n}
\int |\phi(\bs{\omega})|^2 |K(h\bs{\omega})|^2\,d\bs{\omega}
\!\!\!&\leq&\!\!\!
\frac{||K_{\F}||^2_{\infty}}{n}
\int f^2(\bs{x})\,d\bs{x},
\end{eqnarray*}
where 
$||K_{\F}||_{\infty}=\sup_{\bs{\omega}}
\{|K_{\F}(\bs{\omega})|\}$.

\subsection{The upper bound on the Fourier MISE formula}
Concerning an upper bound on the Fourier MISE formula
(\ref{MISEFourier}), we actually sum up the results obtained in
the preceding sections. First of all, to obtain the upper bound we can
omit the correction (the third) term in (\ref{MISEFourier}).
The second term is solely determined by the properties of the kernel, which
is expressed by formula (\ref{2ndMISEterm}). Finally, to obtain
a bound on the first term, the properties of the density the data
are sampled from and the properties of the kernel have to be 
matched somehow. To do this we introduce the so-called Sobolev
class of densities.

\begin{definition}
\label{defSob}
Let $\beta\geq 1$ be an integer and $L>0$.
The Sobolev class of densities $\mathcal{P}_{S({\beta,L})}$
consists of all probability density functions $f:\R^d\ra\R$ satisfying
\begin{equation}
\label{SobolevProp}
\int ||\bs{\omega}||^{2\beta}|\phi(\bs{\omega})|^2\,d\bs{\omega}
\leq (2\pi)^d L^2,
\end{equation}
where $\phi(\bs{\omega})=\F[f](\bs{\omega})$ and $||\cdot||$ is 
the Euclidean norm.
\end{definition}

The condition (\ref{SobolevProp}) is related to the boundedness of
partial derivatives of densities in the Sobolev class; e.g., it can be shown 
that if $\int (\partial{f}/\partial{x_j})^2\,d\bs{x}\leq L_j<\infty$ for all
$j=1,\dots,d$, then (\ref{SobolevProp}) holds for 
$\beta=1$ and $L=||(L_1,\dots,L_d)||$.
Furthermore, if $f\in\mathcal{P}_S(\beta,L)$, 
for some $\beta\in\N$ and $L>0$, then $f\in L_2(\R^d)$.

Now, the announced matching is provided by the fitting
the order of the kernel to the Sobolev character of the 
estimated density. The next theorem, which is the variant of 
Theorem~1.5 in \cite{Tsybakov2009}, provides the final result.

\begin{theorem}
Let $n\in\N$ be the number of i.i.d. samples from a~distribution with 
the density $f:\R^d\ra[0,\infty)$ which is $\beta$-Sobolev for some $\beta\in\N$
and $L>0$, i.e., $f\in\mathcal{P}_S(\beta,L)$. Let $K$ be a kernel of order $\beta$.
Assume that inequality $\mathrm{(\ref{Acondition})}$ holds for some
constant $A>0$. Fix $\alpha>0$ and set $h=\alpha n^{-\frac{1}{2\beta+d}}$.
Then for any $n\geq 1$ the kernel density estimate $\hat{f}_n$ satisfies
\begin{equation}
\sup_{f\in\mathcal{P}_S(\beta,L)}
\E \int (\hat{f}_n(\bs{x})-f(\bs{x}))^2\,d\bs{x}
\leq C\!\cdot\! n^{-\frac{2\beta}{2\beta+d}},
\end{equation}
where $C>0$ is a constant depending only on $\alpha,\beta,d,A,L$ and 
on the kernel $K$.
\end{theorem}

\textbf{Proof.} By Theorem~\ref{bdTheorem} and from the definition of the
Sobolev class of densities, we have
$$
\int |1-K_{\F}(h\bs{\omega})|^2|\phi(\bs{\omega})|^2\,d\bs{\omega}
\leq
(2\pi)^d A^2 h^{2\beta}L^2.
$$
Plugging this into the Fourier MISE formula (\ref{MISEFourier}) 
and employing 
$$\frac{1}{(2\pi)^{d}\,n}\int |K(h\bs{\omega})|^2 d\bs{\omega}=
\frac{1}{nh^d}\int K^2(\bs{u})\,d\bs{u},$$
we get for $h=\alpha n^{-\frac{1}{2\beta+d}}$ the following:
$$
h^{2\beta}=\alpha^{2\beta} n^{-\frac{2\beta}{2\beta+d}},\;
(nh^d)^{-1}=n^{-1}\alpha^{-d}n^{\frac{d}{2\beta+d}}=
\alpha^{-d}n^{-\frac{2\beta}{2\beta+d}}
$$
and
\begin{eqnarray*}
\MISE(\hat{f}_n)\!\!\!&\leq&\!\!\!
\frac{1}{(2\pi)^d}\left[
\int|1-K_{\F}(h\bs{\omega})|^2|\phi(\bs{\omega})|^2\,d\bs{\omega}\;+\;
\frac{1}{n}\int |K_{\F}(h\bs{\omega})|^2\,d\bs{\omega}\right]\\
&\leq&\!\!\!
A^2 h^{2\beta}L^2 +\frac{1}{nh^d}\int K^2(\bs{u})\,d\bs{u},\\
&\leq&\!\!\!
(AL)^2\alpha^{2\beta} n^{-\frac{2\beta}{2\beta+d}}
+\alpha^{-d}n^{-\frac{2\beta}{2\beta+d}}\!\!\int\!\! K^2(\bs{u})\,d\bs{u},\\
&\leq&\!\!\!
\left[(AL)^2\alpha^{2\beta}+\alpha^{-d}\int K^2(\bs{u})\,d\bs{u}\right]
\cdot n^{-\frac{2\beta}{2\beta+d}},\\
&\leq&\!\!\!
C(\alpha,\beta,d,A,L,K)\cdot n^{-\frac{2\beta}{2\beta+d}}.
\end{eqnarray*}
This concludes the proof.\hfill$\Box$\\

The theorem provides the upper bound on the MISE of the multivariate kernel
density estimate (\ref{fnDdim}), if the order of the employed kernel fits to 
the Sobolev character of the density the employed data are sampled from.

\section{Particle filter and kernel methods}
\label{SecIV}
This section presents our own research in the area of the combination of 
the particle filter and kernel methods. The main question here is if the kernel
density estimates constructed on the basis of empirical measures 
approximate the related filtering densities reasonably well. 
The main obstacle to a direct application of the presented kernel estimation
methodology is the fact that the generated empirical measures
are not based on i.i.d. samples due to the resampling step of the filter.

Our results are twofold. First, we show that, despite the mentioned obstacle,
the standard kernel density estimates still converge to the related 
filtering densities. The proof of the assertion is based on Fourier analysis of the
convergence result for the particle filter.

The second result concerns a deeper analysis of the obtained convergence formula.
The convergence result is based on the assumption on the Sobolev character
of the filtering densities. We present a sufficient condition for the persistency of 
this Sobolev character over time.

We extend both results to the partial derivatives of the kernel density estimates
and to the partial derivatives of the filtering densities, respectively.

\subsection{Convergence of kernel density estimates}
To start, let us remind that the particle filter
generates at each time step \mbox{$t=1,\dots ,T$}, $T\in\N$ 
the empirical measure 
$\pi_t^n(d\bs{x}_t)=\frac{1}{n}\sum_{i=1}^n \delta_{\bs{x}^i_t}(d\bs{x}_t)$.
This measure approximates the related filtering distribution $\pi_t$
that is assumed to have the density $p_t(\bs{x}_t)=p(\bs{x}_t|\bs{y}_{1:t})$
with respect to the $d$-dimensional Lebesgue measure, i.e., 
$\pi_t(d\bs{x}_t)=p_t(\bs{x}_t)\,d\bs{x}_t$.

A carrier of the empirical measure $\pi_t^n$ is the set of particles 
$\{\bs{x}^i_t\}_{i=1}^n$, \mbox{$n\in\N$}. This set does not constitute an i.i.d. sample
from $\pi_t$. If one constructs the standard kernel density estimate on the basis
of $\{\bs{x}^i_t\}_{i=1}^n$ and the selected kernel~$K$, i.e., the estimate
\begin{equation}
\label{ptn}
\hat{p}_t^n(\bs{x}_t)=
\frac{1}{nh^d}
\sum_{i=1}^n
K\!\!\left(\frac{\bs{x}_t-\bs{x}^i_t}{h}\right),
\end{equation}
then we ask if $\hat{p}^n_t$ converges in the MISE to the filtering density $p_t$,
provided that the number of particles goes to infinity.

\begin{theorem}
\label{MainTh}
In the filtering problem, 
let $\{\pi_t\}_{t=0}^T$, $\{p_t\}_{t=0}^T$, $T\in\N$ be the sequences of
filtering distributions and corresponding filtering densities. 
Let $p_t$, $t\in\{0,1,\dots,T\}$ be \mbox{$\beta$-Sobolev}
for some $\beta\in\N$ and $L_t>0$, i.e., $p_t\in\mathcal{P}_{S(\beta,L_t)}$.
Let $\{\pi_{t}^n\}_{t=1}^T$, $\{\hat{p}^{n}_t\}_{t=1}^T$, $n\in\N$ be 
the sequences of the empirical measures generated by the particle filter 
and related kernel density estimates (\ref{ptn}) with the bandwidth varying as 
$h(n)=\alpha n^{-\frac{1}{2\beta+d}}$ for some $\alpha>0$.
Let the kernel~$K$ employed in the estimates be of order~$\beta$.
Then we have the following evolution of the MISE of $\hat{p}^{n}_t$ 
over time $t\in\{1,\dots,T\}\mathrm{:}$
\begin{equation}
\label{MThEq1}
\E\left[\int (\hat{p}_t^{n}(\bs{x}_t)-p_t(\bs{x}_t))^2\,d\bs{x}_t\right]
\leq C^2_t\cdot n^{-\frac{2\beta}{2\beta+d}},
\end{equation}
where
\begin{equation}
\label{MThEq2}
C_t=AL_t\alpha^{\beta}+c_t \alpha^{-d/2} ||K||.
\end{equation}\\
In (\ref{MThEq2}), $A$ is the constant of Theorem~\ref{bdTheorem},
$c_t$, $t\in\{1,\dots,T\}$ are the constants of Theorem~\ref{ctTheorem}
and $||K||$ is the $L_2$ norm of the kernel~$K$.
\end{theorem}

\noindent
\textbf{Proof.}
The proof is based on the employment of the Fourier transform.
We start by the assertion of Theorem~\ref{ctTheorem}:
\begin{equation}
\label{ctform}
\E[|\pi_t^nf-\pi_tf|^2]\leq\frac{c^2_t||f||^2_{\infty}}{n},\\
\end{equation}
where we replace a general function $f\in\mathcal{C}^{\C}_b(\mathbb{R}^{d_x})$
by the complex exponential specified on~$\R^d$. Note that $d_x=d$.

Let $f(\bs{x}_t)=e^{\mathrm{i}\lss\bs{\omega},\bs{x}_t\rss}$, then $||f||_{\infty}=1$.
Denoting $\psi^n_t=\F[\pi^n_t]$ and $\psi_t=\F[\pi_t]$
we have from the above
\begin{eqnarray}
\E[|\psi_t^n(\bs{\omega})-\psi_t(\bs{\omega})|^2]
&\leq&
\frac{c^2_t}{n},\nonumber\\
|K_{\F}(h\bs{\omega})|^2\cdot\E[|\psi_t^n(\bs{\omega})-\psi_t(\bs{\omega})|^2]
&\leq&
|K_{\F}(h\bs{\omega})|^2\cdot\frac{c^2_t}{n},\nonumber\\
\E\;[|\psi_t^n(\bs{\omega})K_{\F}(h\bs{\omega})
-\psi_t(\bs{\omega})K_{\F}(h\bs{\omega})|^2]
&\leq&
|K_{\F}(h\bs{\omega})|^2\cdot\frac{c^2_t}{n},\nonumber\\
\E\left[\int |\psi_t^n(\bs{\omega})K_{\F}(h\bs{\omega})
-\psi_t(\bs{\omega})K_{\F}(h\bs{\omega})|^2\,d\bs{\bs{\omega}}\right]
&\leq&
\frac{c^2_t}{n}\!\!\int |K_{\F}(h\bs{\omega})|^2\,d\bs{\omega},\nonumber\\
\E\left[\int (\hat{p}_t^{n}(\bs{x}_t)-p^*_t(\bs{x}_t))^2\,d\bs{x}_t\right]
&\leq&
\frac{c^2_t}{nh^{d}}\!\!\int\!\!K^2(\bs{u})\,d\bs{u}.
\label{E1fce}
\end{eqnarray}
For any density $p_t$ and its convolution \mbox{$p_t^*=p_t*(h^{-d}K(\cdot/h))$},
\begin{eqnarray}
\int (p_t^*(\bs{x}_t)-p_t(\bs{x}_t))^2\,d\bs{x}_t
&=&
\frac{1}{(2\pi)^d}
\int |\psi_t(\bs{\omega})K_{\F}(h\bs{\omega})-\psi_t(\bs{\omega})|^2\,d\bs{\omega}\nonumber\\
&=&
\frac{1}{(2\pi)^d}
\int |1-K_{\F}(h\bs{\omega})|^2|\psi_t(\bs{\omega})|^2\,d\bs{\omega}.
\label{E2fce}
\end{eqnarray}
We assume that the employed kernel has order $\beta$ and 
$p_t\in\mathcal{P}_{S({\beta,L_t})}$. 
Therefore the right-hand side of (\ref{E2fce}) is bounded
according to Theorem~\ref{bdTheorem}. Further, there is nothing random
here and we can apply the expectation with no effect to obtain
\begin{equation}
\label{E3fce}
\E\left[\int (p_t^*(\bs{x}_t)-p_t(\bs{x}_t))^2\,d\bs{x}_t\right]
\leq
A^2 h^{2\beta}L_t^2.
\end{equation}

To proceed, let us consider the product measure $\lambda^d\otimes P$
with the corresponding norm $||\cdot||_{\lambda^d\otimes P}=
[\,\int\!\int |\cdot|^2 d(\lambda^d\otimes P)\,]^{1/2}$.
We have  
\begin{equation}
||\hat{p}_t^{n}(\bs{x}_t)-p_t(\bs{x}_t)||_{\lambda^d\otimes P}
\leq
A h^{\beta}L_t+
\frac{c_t}{(nh^d)^{1/2}}||K||
\label{Etriangle}
\end{equation}
by (\ref{E1fce}), (\ref{E3fce}) and the triangle inequality for
$||\cdot||_{\lambda^d\otimes P}$.

Let the bandwidth $h$ develop with $n$ as $h(n)=\alpha n^{-\frac{1}{2\beta+d}}$
for some $\alpha>0$. We~have
$h^\beta=\alpha^\beta n^{-\frac{\beta}{2\beta+d}}$. Further,
$
(nh^d)^{-1}=n^{-1}\alpha^{-d}n^{\frac{d}{2\beta+d}}=\alpha^{-d}n^{-\frac{2\beta}{2\beta+d}}
$
and therefore
$
(nh^d)^{-1/2}=\alpha^{-d/2}n^{-\frac{\beta}{2\beta+d}}.
$
Inequality (\ref{Etriangle}) then reads as
\begin{eqnarray*}
||\hat{p}_t^{n}(\bs{x}_t)-p_t(\bs{x}_t)||_{\lambda^d\otimes P}&\leq&
AL_t\alpha^{\beta} n^{-\frac{\beta}{2\beta+d}}
+c_t \alpha^{-d/2}n^{-\frac{\beta}{2\beta+d}} ||K||\nonumber\\
&\leq&
(AL_t\alpha^{\beta}+c_t \alpha^{-d/2} ||K||)\cdot
n^{-\frac{\beta}{2\beta+d}}.
\end{eqnarray*}
Squaring to obtain the MISE we get
$$
\E\int(\hat{p}_t^{n}(\bs{x}_t)-p_t(\bs{x}_t))^2\, d{\bs{x}_t}
\leq
(AL_t\alpha^{\beta}+c_t \alpha^{-d/2} ||K||)^2\cdot
n^{-\frac{2\beta}{2\beta+d}}
$$
or in the more compact form
$$
\E\int (\hat{p}_t^{n}(\bs{x}_t)-p_t(\bs{x}_t))^2\,d{\bs{x}_t}
\leq C_t^2\cdot n^{-\frac{2\beta}{2\beta+d}}
$$
for
$C_t=AL_t\alpha^{\beta}+c_t \alpha^{-d/2}||K||$.
\hfill$\Box$\\

Let us discuss the theorem.\\

1) First of all, the theorem is proved without any assumption on the 
i.i.d.~character of samples (particles) constituting the empirical measures 
$\pi^n_{t}$. This is the crucial observation, as we know that due to the 
resampling step  the generated particles are not i.i.d.

2) Convergence. For $t\in\N$ fixed, we immediately see from
(\ref{MThEq1}) that the~MISE of kernel estimates goes to zero as the number of
particles increases and the bandwidth decreases accordingly,
i.e., $\lim_{n\ra\infty} \E\int (\hat{p}_t^{n}(\bs{x}_t)-p_t(\bs{x}_t))^2\,d{\bs{x}_t}=0$.

3) Consistency. The theorem proposes that the bandwidth develops with the
number of particles $n$ as $h(n)=\alpha n^{-\frac{1}{2\beta+d}}$ for some 
$\alpha>0,\beta,d\in\N$. Obviously, $\lim_{n\ra\infty} h(n)=0$, and
$\lim_{n\ra\infty} nh(n)=\lim_{n\ra\infty} 
\alpha n^{\frac{2\beta+d-1}{2\beta+d}}=\infty$.

4) The dimension matters. We have $n^{-\frac{2\beta}{2\beta+d_1}}<
n^{-\frac{2\beta}{2\beta+d_2}}$ for $d_1<d_2$, and therefore we must 
increase the number of particles in order to assure a given accuracy as
the dimension increases.

5) The order helps. Contrary to the previous result, we have
$n^{-\frac{2\beta_1}{2\beta_1+d}}>n^{-\frac{2\beta_2}{2\beta_2+d}}$
for $\beta_1<\beta_2$.
Hence the greater is the order of the employed kernel, the tighter is the
bound on the related MISE, in fact, it tends towards $n^{-1}$.
There are techniques available for constructing kernels of arbitrary
orders \cite{Tsybakov2009}, however, the order of the 
employed kernel is primarily driven by the Sobolev character 
of the filtering densities.

6) The theorem assumes that the filtering densities $p_t$ are $\beta$-Sobolev
for some~$L_t>0$, $t\in\{0,\dots,T\}$, $T\in\N$ and $\beta\in\N$ being 
constant over time.
It is the question when this assumption holds. In Section~\ref{SecV}, we show 
that the Sobolev character of the filtering densities is retained over time,
if a certain condition holds on the transition kernels of the signal process.

7) For $\alpha=1$, the specification of $C_t$ simplifies to $C_t=AL_t+c_t||K||$
and $C_t$ consists of four terms. Two of them, $A$ and 
$||K||=[\int K^2(\bs{u})\,d\bs{u}]^{1/2}$ are the constants determined by
the employed kernel. The other two, $L_t$ and $c_t$, develop with time.
The $L_t$ term is discussed in Section~\ref{SecV}.

8) The $c_t$ constant (with respect to the number of particles)
comes from Theorem~\ref{ctTheorem}. It can be shown that its
values can be computed recursively as 
$c_t=c_{t-1}\left(1+\frac{4||g^v_t||_{\infty}}{\ov{\pi}_tg_t}\right), c_0=1$.
The integral $\ov{\pi}_tg_t$ depends on the values of the observation
process and $c_t$ generally develops exponentially with time, 
see the remark in concluding Section~\ref{SecVII}.

\subsection{Extension to partial derivatives}
The result of Theorem~\ref{MainTh} can be straightforwardly extended
to the convergence of partial derivatives of kernel density estimates
to partial derivatives of the filtering densities. The proof of the assertion
substantially overlaps with the proof of Theorem~\ref{MainTh}, however,
we present it here in full detail for the convenience of the reader.

In what follows we denote by $p_{t,i_1,\dots,i_d}^{(m)}=\partial^m p_t/
\partial x^{i_1}_1\dots\partial x^{i_d}_d$ the $m$-th partial
derivative of the filtering density $p_t$, $t\in\{0,\dots,T\}$, $T\in\N$
for \mbox{$i_1,\dots,i_d\in\N_0$}, such that $m=i_1+\dots+i_d$,
and $m\in\N_0$. Similarly, we will use $\hat{p}^{n,(m)}_{t,i_1,\dots,i_d}=
\partial^m \hat{p}^n_t/\partial x^{i_1}_1\dots\partial x^{i_d}_d$
for the partial derivative of kernel estimate (\ref{ptn}), 
$p^{*,(m)}_{t,i_1,\dots,i_d}$ for the partial derivative of the convolution
$p_t^*=p_{t}*(h^{-d}K(\cdot/h))$ and
$K^{(m)}_{i_1,\dots,i_d}=\partial^m K/
\partial u^{i_1}_1\dots\partial u^{i_d}_d$ for the partial derivative
of the kernel employed in the estimates.
Clearly, the zero value of $i_j$, $j=1,\dots, d$, 
corresponds to the situation when no differentiation is applied 
in the respective dimension.

\begin{theorem}
\label{PdMainTh}
In the filtering problem,
let $\{\pi_t\}_{t=0}^T$, $\{p_{t,i_1,\dots,i_d}^{(m)}\}_{t=0}^T$, 
$T\in\N$, \mbox{$m\in\N_0$} be the sequences of filtering distributions 
and $m$-th partial derivatives of corresponding filtering densities
for some $i_1,\dots,i_d\in\N_0$, $m=i_1+\dots+i_d$.
Let $p_{t,i_1,\dots,i_d}^{(m)}$, $t\in\{0,\dots,T\}$ satisfy
(\ref{SobolevProp}) for some $\beta\in\N$ and $L_{t,(m)}>0$.
Let $\{\pi_{t}^n\}_{t=1}^T$, $\{\hat{p}^{(m)}_{t,i_1,\dots,i_d}\}_{t=1}^T$, 
$n\in\N$ be the sequences of the empirical measures generated by 
the particle filter and $m$-th partial derivatives of the related kernel density 
estimates (\ref{ptn}) with the bandwidth varying as 
$h(n)=\alpha n^{-\frac{1}{2\beta+d+2m}}$ for some $\alpha>0$.
Let the kernel~$K$ employed in the estimates be of order~$\beta$.
Then we have the following evolution of the MISE of the $m$-th partial
derivatives of kernel estimates
$\hat{p}^{n,(m)}_{t,i_1,\dots,i_d}$ over time $t\in\{1,\dots,T\}\mathrm{:}$
\begin{equation}
\label{PdMThEq1}
\E\left[\int (\hat{p}^{n,(m)}_{t,i_1,\dots,i_d}(\bs{x}_t)-
p^{(m)}_{t,i_1,\dots,i_d}(\bs{x}_t))^2\,d\bs{x}_t\right]
\leq C^2_{t,(m)}\cdot n^{-\frac{2\beta}{2\beta+d+2m}},
\end{equation}
where
\begin{equation}
\label{PdMThEq2}
C_{t,(m)}=AL_{t,(m)}\alpha^{\beta}+c_t \alpha^{-(d/2+m)}||K^{(m)}_{i_1,\dots,i_d}||.
\end{equation}\\
In (\ref{PdMThEq2}), $A$ is the constant of Theorem~\ref{bdTheorem},
$c_t$, $t\in\{1,\dots,T\}$ are the constants of Theorem~\ref{ctTheorem}
and $||K^{(m)}_{i_1,\dots,i_d}||$ is the $L_2$ norm of the corresponding
$m$-th partial derivative of kernel~$K$.
\end{theorem}

\noindent
\textbf{Proof.}
To start remind that for any function $f:\R^d\ra \R$ and its $m$-th partial derivative 
$f_{i_1,\dots,i_d}^{(m)}:\R^d\ra\R$,  both assumed in $L_1(\R^d)$, 
one has for their Fourier transforms  $\F[p](\bs{\omega})$ and 
$\F[p_{i_1,\dots,i_d}^{(m)}](\bs{\omega})$, respectively, the equality
\begin{equation}
\F[f^{(m)}_{i_1,\dots,i_d}](\bs{\omega})=
(\mathrm{-i})^m
(\omega^{i_1}_{i_1}\cdot\dots\cdot\omega^{i_d}_{i_d})
\F[f](\bs{\omega}).
\label{Fder}
\end{equation}

Now, in order to prove the theorem, we just mimic the proof 
of Theorem~\ref{MainTh}. Employing the complex exponential
in (\ref{ctform}) and the equality (\ref{Fder}) we have
\begin{eqnarray}
\hspace{-1cm}&&\hspace{1.4cm}
\E[|\psi_t^n(\bs{\omega})-\psi_t(\bs{\omega})|^2]\;\leq\;\frac{c^2_t}{n},\nonumber\\
\hspace{-1cm}&&|(\mathrm{-i})^m(\omega^{i_1}_{1}\dots\cdot\omega^{i_d}_{d})
K_{\F}(h\bs{\omega})|^2\cdot\E[|\psi_t^n(\bs{\omega})-\psi_t(\bs{\omega})|^2]\nonumber\\
&&\hspace{4.55cm}
\;\leq\;|(\mathrm{-i})^m(\omega^{i_1}_{1}\dots\cdot\omega^{i_d}_d)
K_{\F}(h\bs{\omega})|^2\cdot\frac{c^2_t}{n},\nonumber\\
\hspace{-1cm}&&\E\;[|(\mathrm{-i})^m(\omega^{i_1}_{1}\dots\cdot\omega^{i_d}_{d})
(\psi_t^n(\bs{\omega})K_{\F}(h\bs{\omega}))
-(\mathrm{-i})^m(\omega^{i_1}_{1}\dots\cdot\omega^{i_d}_{d})
(\psi_t(\bs{\omega})K_{\F}(h\bs{\omega}))|^2]\nonumber\\
\hspace{-1cm}&&\hspace{4.55cm}
\;\leq\;|(\mathrm{-i})^m(\omega^{i_1}_{1}\dots\cdot\omega^{i_d}_{d})
K_{\F}(h\bs{\omega})|^2\cdot\frac{c^2_t}{n},\nonumber\\
\hspace{-1cm}&&
\E\left[\int |\F[\,\partial^m \hat{p}^n_t/\partial x^{i_1}_1\dots\partial x^{i_d}_d](\bs{\omega})
-\F[\,\partial^m p^*_t/\partial x^{i_1}_1\dots\partial x^{i_d}_d](\bs{\omega})|^2\,
d\bs{\bs{\omega}}\right]\nonumber\\
\hspace{-1cm}&&\hspace{4.55cm}
\;\leq\;\frac{c^2_t}{n}\!\!
\int |\F[\partial^m h^{-d}K(\bs{x}/h)/
\partial x^{i_1}_1\dots\partial x^{i_d}_d]|^2\,d\bs{\omega},\nonumber\\
\hspace{-1cm}&&
\E\left[\int |\F[\,\hat{p}_{t,i_1,\dots,i_d}^{n,(m)}(\bs{x}_t)](\bs{\omega})
-\F[\,p_{t,i_1,\dots,i_d}^{*,(m)}(\bs{x}_t)](\bs{\omega})|^2\,
d\bs{\bs{\omega}}\right]\nonumber\\
\hspace{-1cm}&&\hspace{4.55cm}
\;\leq\;\frac{c^2_t}{n}\!\!
\int |\F[\partial^m h^{-d}K(\bs{x}/h)/
\partial x^{i_1}_1\dots\partial x^{i_d}_d]|^2\,d\bs{\omega},\nonumber\\
\hspace{-1cm}&&
\E\left[\int (\hat{p}_{t,i_1,\dots,i_d}^{n,(m)}(\bs{x}_t)
-p_{t,i_1,\dots,i_d}^{*,(m)}(\bs{x}_t))^2\,
d\bs{x}_t\right]\nonumber\\
\hspace{-1cm}&&\hspace{4.55cm}
\;\leq\;\frac{c^2_t}{nh^{2d}}\!\!
\int (\partial^m K(\bs{x}/h)/
\partial x^{i_1}_1\dots\partial x^{i_d}_d)^2\,d\bs{x},\nonumber\\
\hspace{-1cm}&&
\E\left[\int (\hat{p}_{t,i_1,\dots,i_d}^{n,(m)}(\bs{x}_t)
-p_{t,i_1,\dots,i_d}^{*,(m)}(\bs{x}_t))^2\,
d\bs{x}_t\right]\nonumber\\
\hspace{-1cm}&&\hspace{4.55cm}
\;\leq\;\frac{c^2_t}{nh^{d+2m}}\!\!
\int (\partial^m K(\bs{u})/
\partial u^{i_1}_1\dots\partial u^{i_d}_d)^2\,d\bs{u},\nonumber\\
\hspace{-1cm}&&\E\left[\int (\hat{p}_{t,i_1,\dots,i_d}^{n,(m)}(\bs{x}_t)-
p_{t,i_1,\dots,i_d}^{*,(m)}(\bs{x}_t))^2\,d\bs{x}_t\right]\nonumber\\
\hspace{-1cm}&&\hspace{4.55cm}
\;\leq\;\frac{c^2_t}{nh^{d+2m}}
\int (K^{(m)}_{i_1,\dots,i_d}(\bs{u}))^2\,d\bs{u}.\nonumber
\end{eqnarray}

Using the $L_2$ norm of $K^{(m)}_{i_1,\dots,i_d}(\bs{u})$, i.e.,
$||K^{(m)}_{i_1,\dots,i_d}||^2=
\int (K^{(m)}_{i_1,\dots,i_d}(\bs{u}))^2\,d\bs{u}$, the above reads as
\begin{equation}
\E\!\int (\hat{p}_{t,i_1,\dots,i_d}^{n,(m)}(\bs{x}_t)-
p^{*,(m)}_{t,i_1,\dots,i_d}(\bs{x}_t))^2\,d{\bs{x}_t}
\leq \frac{c^2_t}{nh^{d+2m}}\,||K_{i_1,\dots,i_d}^{(m)}||^2.
\label{PdE1fce}
\end{equation}

For given $i_1,\dots i_d\in\N_0$, we assume that $p_{t,i_1,\dots,i_d}^{(m)}$, 
$t\in\{0,\dots,T\}$ exist and are \mbox{$\beta\in\N$} Sobolev in the sense of
validity of (\ref{SobolevProp}). That is, for 
the Fourier transforms $\F[p_{t,i_1,\dots,i_d}^{(m)}](\bs{\omega})$ there
exist positive constants $L_{t,(m)}$ such that
\begin{equation}
\label{PdSobolev}
\int ||\bs{\omega}||^{2\beta} |\F[\,p_{t,i_1,\dots,i_d}^{(m)}](\bs{\omega})|^2
\,d\bs{\omega}\,\leq\, (2\pi)^d L^{2}_{t,(m)}.
\end{equation}
Using (\ref{PdSobolev}) we have under the assumptions of
Theorem~\ref{bdTheorem} the formula
\begin{equation}
\label{conv1Kd}
\int\! |1-K_{\F}(h\bs{\omega})|^2
|\F[p_{t,i_1,\dots,i_d}^{(m)}](\bs{\omega})|^2\,d\bs{\omega}
\leq
(2\pi)^d A^2 h^{2\beta}L_{t,(m)}^2.
\end{equation}

By (\ref{conv1Kd}) we get the counterpart of (\ref{E2fce}) that writes as
\begin{eqnarray*}
&&\hspace{-1cm}\int (p_{t,i_1,\dots,i_d}^{*,(m)}(\bs{x}_t)-
p_{t,i_1,\dots,i_d}^{(m)}(\bs{x}_t))^2\,d\bs{x}_t\nonumber\\
&&\hspace{2.5cm}=\frac{1}{(2\pi)^d}
\int |1-K_{\F}(h\bs{\omega})|^2|
\F[p^{(m)}_{t,i_1,\dots,i_d}](\bs{\omega})|^2\,d\bs{\omega}\\[0.2cm]
&&\hspace{2.5cm}\leq A^2 h^{2\beta}L_{t,{(m)}}^2.
\end{eqnarray*}

We proceed in the same way as in the proof of Theorem~\ref{MainTh}. 
We consider the $||\cdot||_{\lambda^d\otimes P}$ norm and employ
the triangle inequality to get
\begin{equation}
||\hat{p}_{t,i_1,\dots,i_d}^{n,(m)}(\bs{x}_t)-
p_{t,i_1,\dots,i_d}^{(m)}(\bs{x}_t)||_{\lambda^d\otimes P}
\leq
A h^{\beta}L_{t,(m)}+
\frac{c_t}{(nh^{d+2m})^{-1/2}}||K_{i_1,\dots,i_d}^{(m)}||.
\label{PdTriangle}
\end{equation}

The bandwidth $h$ develop with $n$ as $h(n)=\alpha n^{-\frac{1}{2\beta+d+2m}}$
for some $\alpha>0$. So we have
$h^\beta=\alpha^\beta n^{-\frac{\beta}{2\beta+d+2m}}$. Further,
$
(nh^{d+2m})^{-1}=n^{-1}\alpha^{-(d+2m)}n^{\frac{d+2m}{2\beta+d+2m}}
=\alpha^{-(d+2m)}n^{-\frac{2\beta}{2\beta+d+2m}}
$
and therefore
$
(nh^{d+2m})^{-1/2}=\alpha^{-(d+2m)/2}n^{-\frac{\beta}{2\beta+d+2m}}.
$
This gives us after squaring (\ref{PdTriangle}) the statement of
the theorem:
\begin{equation}
\label{ctpder}
\E\int (\hat{p}_{t,i_1,\dots,i_m}^{n,(m)}(\bs{x}_t)-
p^{(m)}_{t,i_1,\dots,i_m}(\bs{x}_t))^2\,d{\bs{x}_t}
\leq C_{t,(m)}^2\cdot n^{-\frac{2\beta}{2\beta+d+2m}}
\end{equation}
for
$C_{t,(m)}=AL_{t,(m)}\alpha^{\beta}+
c_t \alpha^{-(d+2m)/2}||K_{i_1,\dots,i_d}^{(m)}||$.
\hfill$\Box$\\

The structure of formula (\ref{ctpder}) is the same as that of formula 
(\ref{MThEq1}) of Theorem~\ref{MainTh}. Only two constants are replaced.
Therefore, the discussion of its corollaries remains valid, especially,
it implies the convergence of partial derivatives of the kernel density
estimates to the respective derivatives of the related filtering densities.

On the other hand, we see that the order of the partial derivative $m$ 
slows down the convergence. In fact, it has the same effect on the
convergence as the dimension $d$, see the discussion concerning
the influence of the dimension below Theorem~\ref{MainTh}.

\section{Sobolev character of filtering densities}
\label{SecV}
In Theorem \ref{MainTh}, we have assumed that the filtering densities
$p_t$, $t\in\{0,\dots,T\}$, $T\in\N$ are $\beta$-Sobolev over time. 
This assumption can be verified for $p_0$, but for other time instants $t>0$
a~direct verification is typically impossible. That is why we are interested
in a practical tool for performing the verification indirectly so that the assumptions
for the convergence result of Theorem~\ref{MainTh} were fulfilled.
As a result, we present a sufficient condition on the densities
of transition kernels of the signal process such that the Sobolev character
of the filtering densities is retained over time.

In the statement below, we work with the prediction and update formulas,
(\ref{L1}) and (\ref{L2}), respectively, of Section \ref{FeSec}. 
We rewrite these formulas in the more compact form using the following shortcuts: 
$\ov{p}_t(\bs{x}_t)=p(\bs{x}_t|\bs{y}_{1:t-1})$,
$p_t(\bs{x}_t)=p(\bs{x}_t|\bs{y}_{1:t})$ (in fact, this shortcut was already used in
Theorem~\ref{MainTh}) and $g_t(\bs{x}_t)=g_t(\bs{y}_t|\bs{x}_t)$ 
for the respective densities; and 
$\ov{\pi}_{t}g_t={\int g_t(\bs{y}_t|\bs{x}_t)p(\bs{x}_{t}|\bs{y}_{1:t-1})\,d\bs{x}_t}
={\int g_t(\bs{x}_t)\ov{p}_t(\bs{x}_{t})\,d\bs{x}_t}$ for the normalizing integral.
Using the introduced shortcuts we have (\ref{L1}) and (\ref{L2}) written as
\begin{eqnarray}
\label{PredFT}
\ov{p}_t(\bs{x}_t)&=&
\int K_{t-1}(\bs{x}_t|\bs{x}_{t-1})p_{t-1}(\bs{x}_{t-1})\,d\bs{x}_{t-1},\\
p_t(\bs{x}_t)&=&\frac{g_t(\bs{x}_t)\ov{p}_t(\bs{x}_t)}{\ov{\pi}_tg_t}.
\label{UpdateFT}
\end{eqnarray}

\begin{definition}
Let $K_{t-1}$ be the transition kernel in the filtering problem 
for time $t-1$, $t-1\in\N_0$.  As the conditional characteristic function 
$\F[K_{t-1}](\bs{\omega}|\bs{x}_{t-1})$
of the transition kernel $K_{t-1}$ we denote the characteristic function
of the conditional distribution determined by this kernel,~i.e.,
$$
\F[K_{t-1}](\bs{\omega}|\bs{x}_{t-1})=
\int e^{\mathrm{i}\lss\bs{\omega},\bs{x}_t\rss}
K_{t-1}(d\bs{x}_t|\bs{x}_{t-1}).
$$
\end{definition}

\begin{theorem}
\label{Th2}
In the filtering problem, let $p_0\in\mathcal{P}_{S(\beta,L_0)}$. Let
$K_{t-1}$, $t\in\N$ be the set of the transition kernels and
$\F[K_{t-1}]$, $t\in\N$ be the set of the corresponding conditional
characteristic functions. For all $t\in\N$, let $\F[K_{t-1}]$ be bounded 
by a~function $K_b\!:\R^d\ra\C$
in such a way that for any $\bs{x}_{t-1}\in\R^d$ and $\bs{\omega}\in\R^d$
\begin{equation}
\label{iKF}
|\F[K_{t-1}](\bs{\omega}|\bs{x}_{t-1})|\leq |K_{b}(\bs{\omega})|.
\end{equation}
Let the function $K_b$ satisfy (\ref{SobolevProp})
for some $\beta\in\N$ and $L_{K_b}>0$.
Then the filtering densities $p_t$ are $\beta$-Sobolev for 
all \mbox{$t\in\N$}, i.e., 
$p_t\in\mathcal{P}_{S(\beta,L_t)}$, with the recurrence for $L_t$
written as
\begin{equation}
\label{Ltspec}
L_t=||g^v_t||_{\infty}L_{K_b}/\ov{\pi}_tg_t,
\end{equation}
where $||g^v_t||_{\infty}=\sup_{\bs{u}}\{|g^v_t(\bs{u})|\}$.
\end{theorem}

\noindent
\textbf{Proof.}
The theorem holds for $p_0$ by the assumption.
Let $t\in\N$, then by multiplying both sides of (\ref{PredFT}) 
by the complex exponential we get from the prediction formula 
$$
e^{\mathrm{i}\lss\bs{\omega},\bs{x}_t\rss}\,\ov{p}_t(\bs{x}_t)=
e^{\mathrm{i}\lss \bs{\omega},\bs{x}_t\rss}\!\!
\int\! K_{t-1}(\bs{x}_t|\bs{x}_{t-1})p_{t-1}(\bs{x}_{t-1})\,d\bs{x}_{t-1}.
$$
By integration, the left-hand side gives the characteristic function
$\ov{\psi}_t(\bs{\omega})$ of~$\ov{p}_t(\bs{x}_t)$, i.e.,
$$
\ov{\psi}_t(\bs{\omega})=
\int e^{\mathrm{i}\lss\bs{\omega},\bs{x}_t\rss}\ov{p}_t(\bs{x}_t)\,d\bs{x}_t.
$$
The right-hand side has then form
\begin{eqnarray*}
&&
\!\!\!\int\! \int e^{\mathrm{i}\lss \bs{\omega},\bs{x}_t\rss}
K_{t-1}(\bs{x}_t|\bs{x}_{t-1})p_{t-1}(\bs{x}_{t-1})\,d\bs{x}_{t-1}\,d\bs{x}_t\\
&=&
\!\!\!\int p_{t-1}(\bs{x}_{t-1})\left(
\int e^{\mathrm{i}\lss\bs{\omega},\bs{x}_t\rss}K_{t-1}(\bs{x}_t|\bs{x}_{t-1})\,d\bs{x}_t
\right)d\bs{x}_{t-1},\\
&=&
\!\!\!\int p_{t-1}(\bs{x}_{t-1})\F[K_{t-1}](\bs{\omega}|\bs{x}_{t-1})\,d\bs{x}_{t-1}.
\end{eqnarray*}
The equality of two complex numbers is equivalent to the equality of
their complex conjugates. Hence we can multiply both sides
by their complex conjugates with the equality retained. This gives
us the expression
$$
|\ov{\psi}_t(\bs{\omega})|^2=
\left|\int p_{t-1}(\bs{x}_{t-1})
\F[K_{t-1}](\bs{\omega}|\bs{x}_{t-1})\,d\bs{x}_{t-1}\right|^2.
$$

By the Jensen's inequality and assumed boundedness 
of $\F[K_{t-1}]$, we have
\begin{eqnarray*}
|\ov{\psi}_t(\bs{\omega})|^2&\leq&
\left(\int |\F[K_{t-1}](\bs{\omega}|\bs{x}_{t-1})|\,
p_{t-1}(\bs{x}_{t-1})\,d\bs{x}_{t-1}\right)^{\!\!2}\\
&\leq&
\left(|K_b(\bs{\omega})|
\int\! p_{t-1}(\bs{x}_{t-1})\,d\bs{x}_{t-1}\!\right)^{\!\!2}
\!=|K_b(\bs{\omega})|^2.
\end{eqnarray*}
Thus,
\begin{equation}
\int ||\bs{\omega}||^{2\beta}
|\ov{\psi}_t(\bs{\omega})|^2\,d\bs{\omega}
\leq
\int ||\bs{\omega}||^{2\beta} 
|K_b(\bs{\omega})|^2\leq (2\pi)^d L^2_{K_b}.
\label{Pt1}
\end{equation}

The above formula shows that $\ov{p}_t\in\mathcal{P}_{(\beta,L_{K_b})}$
for any $t\in\N$. We proceed with the specification of the Sobolev 
constant $L_t$ of the update (filtering) density $p_t$.

In Section~\ref{ftds}, in formula (\ref{gtdef}), there was shown 
that the function $g_t(\bs{x}_t)$ of the update formula (\ref{UpdateFT})
has form $g_t(\bs{x}_t)=g^v_t(\bs{y}_t-h(\bs{x}_t))$.
Function $g^v_t$ is the density of the noise term in the observation
process and is assumed to be bounded. Thus,
we have 
$\sup_{\bs{x}_t,\bs{y}_t} \{|g^v_t(\bs{y}_t-h(\bs{x}_t))|\}\!=
\sup_{\bs{u}} \{|g^v_t(\bs{u})|\}=||g^v_t||_{\infty}<\infty$.

Again, multiplying the update formula (\ref{UpdateFT}) 
by the complex exponential, integrating
and multiplying by the respective conjugates gives us
\begin{eqnarray*}
(\ov{\pi}_tg_t)\,p_t(\bs{x}_t)&=&g_t(\bs{x}_t)\,\ov{p}_t(\bs{x}_t),\nonumber\\
(\ov{\pi}_tg_t)\!\!
\int\!\! e^{\mathrm{i}\lss\bs{\omega},\bs{x}_t\rss}p_t(\bs{x}_t)\,d\bs{x}_t
&=&
\!\!\!\int\!\! e^{\mathrm{i}\lss\bs{\omega},\bs{x}_t\rss}
g_t(\bs{x}_t)\,\ov{p}_t(\bs{x}_t)\,d\bs{x}_t,\nonumber\\
(\ov{\pi}_tg_t)^2|\psi_t(\bs{\omega})|^2
&\leq& ||g^v_t||_{\infty}^2|\ov{\psi}_t(\bs{\omega})|^2,\nonumber\\
||\bs{\omega}||^{2\beta}|\psi_t(\bs{\omega})|^2
&\leq& \frac{||g^v_t||_{\infty}^2}{(\ov{\pi}_tg_t)^2}\,||\bs{\omega}||^{2\beta}|
\ov{\psi}_t(\bs{\omega})|^2,\nonumber\\
(2\pi)^{-d}\!\!\!\int\! ||\bs{\omega}||^{2\beta}|\psi_t(\bs{\omega})|^2\,d\bs{\omega}
&\leq&
\frac{||g^v_t||_{\infty}^2L^2_{K_b}}{(\ov{\pi}_tg_t)^2}=L^2_{t}.
\label{Ltdef}
\end{eqnarray*}
This concludes the proof.\hfill$\Box$\\

The theorem tells us that, in the particle filter, the \mbox{$\beta$-Sobolev}
character of the filtering densities is retained over time
if the set of the conditional characteristic functions of transition kernels
$\F[K_{t-1}](\bs{\omega}|\bs{x}_{t-1})$, $t\in\N$ is uniformly bounded.

\subsection{Extension to partial derivatives}
Considering preservation of the Sobolev character of partial derivatives (in
the sense of validity of (\ref{SobolevProp})) of the filtering densities 
$p_{t,i_1,\dots,i_m}^{(m)}$, the theorem holds as well.
The difference is that we assume that $p_{0,i_1,\dots,i_m}^{(m)}$
is \mbox{$\beta$-Sobolev}\footnote{Strictly speaking,
we cannot say that $p_{0,i_1,\dots,i_m}^{(m)}$ is $\beta$-Sobolev
or write $p_{0,i_1,\dots,i_m}^{(m)}\in\mathcal{P}_{(\beta,L_{0,(m)})}$
as the partial derivative is not a density anymore. But, if we still 
do it for a~general function, then we mean that the Fourier transform
of this function exists and satisfies the inequality (\ref{SobolevProp}) for some 
$\beta\in\N$ and $L>0$.} and, in (\ref{iKF}), instead of considering
boundedness of $\F[K_{t-1}](\bs{\omega}|\bs{x}_{t-1})$, we consider
the boundedness of $\F[K_{t-1,i_1,\dots,i_m}^{(m)}](\bs{\omega}|\bs{x}_{t-1})$
for any $\bs{x}_{t-1}\in\R^d$, $t\in\N$.

\begin{theorem}
\label{derTh2}
In the filtering problem, let $p_{0,i_1,\dots,i_m}^{(m)}$ 
be $\beta$-Sobolev for some $\beta\in\N$, $L_{0,(m)>0}$ and
$i_1,\dots,i_d\in\N_0$ such that $m=i_1+\dots+i_d$, $m\in\N_0$.
Let $K_{t-1}$, $t\in\N$ be the set of the transition kernels, and
$K_{t-1,i_1,\dots,i_m}^{(m)}=
\partial^m K_{t-1}/\partial x^{i_1}_1\dots\partial x^{i_d}_d$, $t\in\N$
the set of corresponding partial derivatives. Let 
$\F[K_{t-1,i_1,\dots,i_m}^{(m)}](\bs{\omega}|\bs{x}_{t-1})$, $t\in\N$ be
the set of the corresponding conditional Fourier transforms, i.e.,
$$
\F[K_{t-1,i_1,\dots,i_m}^{(m)}](\bs{\omega}|\bs{x}_{t-1})=
\int e^{\mathrm{i}\lss\bs{\omega},\bs{x}_t\rss}
K_{t-1,i_1,\dots,i_m}^{(m)}(\bs{x}_t|\bs{x}_{t-1})\,d\bs{x}_t.
$$
For all $t\in\N$, let $\F[K_{t-1,i_1,\dots,i_m}^{(m)}]$ be bounded 
by some function $K_{b,(m)}\!:\R^d\ra\C$
in such a way that for any $\bs{x}_{t-1}\in\R^d$ and $\bs{\omega}\in\R^d$,
\begin{equation}
\label{deriKF}
|\F[K_{t-1,i_1,\dots,i_m}^{(m)}](\bs{\omega}|\bs{x}_{t-1})|\leq |K_{b,(m)}(\bs{\omega})|.
\end{equation}
Let the function $K_{b,(m)}$ satisfy (\ref{SobolevProp})
for the above $\beta\in\N$ and some $L_{K_{b,(m)}}>0$.
Then the partial derivatives of filtering densities $p_{t,i_1,\dots,i_m}^{(m)}$,
\mbox{$t\in\N$} are $\beta$-Sobolev with the recurrence for $L_t$
written as
\begin{equation}
\label{derLtspec}
L_t=||g_t||_{\infty}L_{K_{b,(m)}}/\ov{\pi}_tg_t,
\end{equation}
where $||g_t||_{\infty}=\sup_{\bs{u}}\{|g_t(\bs{u})|\}$.
\end{theorem}

\noindent
\textbf{Proof.} 
The theorem holds for $p_{0,i_1,\dots,i_m}^{(m)}$ by the assumption.
From the prediction formula, multiplying both sides of
(\ref{PredFT}) by the complex exponential, we get
$$
e^{\mathrm{i}\lss\bs{\omega},\bs{x}_t\rss}\,\ov{p}_t(\bs{x}_t)=
e^{\mathrm{i}\lss \bs{\omega},\bs{x}_t\rss}\!\!
\int K_{t-1}(\bs{x}_t|\bs{x}_{t-1})p_{t-1}(\bs{x}_{t-1})\,d\bs{x}_{t-1}.
$$
By integration, the left-hand side just gives the characteristic function
$\ov{\psi}_t(\bs{\omega})$ of~$\ov{p}_t(\bs{x}_t)$, i.e.,
$$
\ov{\psi}_t(\bs{\omega})=
\int e^{\mathrm{i}\lss\bs{\omega},\bs{x}_t\rss}\ov{p}_t(\bs{x}_t)\,d\bs{x}_t.
$$
The right-hand side has then form
\begin{eqnarray*}
&&
\int \int e^{\mathrm{i}\lss \bs{\omega},\bs{x}_t\rss}
K_{t-1}(\bs{x}_t|\bs{x}_{t-1})p_{t-1}(\bs{x}_{t-1})\,d\bs{x}_{t-1}\,d\bs{x}_t\\
&=&
\int \!\!p_{t-1}(\bs{x}_{t-1})\left(
\int e^{\mathrm{i}\lss\bs{\omega},\bs{x}_t\rss}K_{t-1}(\bs{x}_t|\bs{x}_{t-1})\,d\bs{x}_t
\right)d\bs{x}_{t-1},\\
&=&
\int p_{t-1}(\bs{x}_{t-1})\F[K_{t-1}](\bs{\omega}|\bs{x}_{t-1})\,d\bs{x}_{t-1}.
\end{eqnarray*}
Multiplying both sides by 
$(-\mathrm{i})^m(\omega^{i_1}_{i_1}\cdot\dots\cdot\omega^{i_d}_{i_d})$, 
we move both sides to the Fourier transforms of the corresponding 
partial derivatives. That is,
$$
(-\mathrm{i})^m(\omega^{i_1}_{i_1}\cdot\dots\cdot\omega^{i_d}_{i_d})
\ov{\psi}_t(\bs{\omega})=\F[\,\ov{p}_{t,i_1,\dots,i_m}^{(m)}]
$$
and
$$
\hspace{-2cm}
\int p_{t-1}(\bs{x}_{t-1})
(-\mathrm{i})^m(\omega^{i_1}_{i_1}\cdot\dots\cdot\omega^{i_d}_{i_d})
\F[K_{t-1}](\bs{\omega}|\bs{x}_{t-1})\,d\bs{x}_{t-1}=
$$
$$
\hspace{4cm}
=\int p_{t-1}(\bs{x}_{t-1})\F[K_{t-1,i_1,\dots,i_m}^{(m)}]
(\bs{\omega}|\bs{x}_{t-1})\,d\bs{x}_{t-1}.
$$
Further multiplying both sides by the complex conjugates gives the
expression
$$
|\F[\,\ov{p}_{t,i_1,\dots,i_m}^{(m)}]|^2=
\left|\int p_{t-1}(\bs{x}_{t-1})
\F[K_{t-1,i_1,\dots,i_m}^{(m)}]
(\bs{\omega}|\bs{x}_{t-1})\,d\bs{x}_{t-1}\right|^2.
$$

Now, by the assumed boundedness of $\F[K_{t-1,i_1,\dots,i_m}^{(m)}]$ and
the Jensen's inequality, we have
\begin{eqnarray*}
|\F[\,\ov{p}_{t,i_1,\dots,i_m}^{(m)}]|^2&\leq&
\left(\int |\F[K_{t-1,i_1,\dots,i_m}^{(m)}](\bs{\omega}|\bs{x}_{t-1})|\,
p_{t-1}(\bs{x}_{t-1})\,d\bs{x}_{t-1}\right)^{\!\!2}\\
&\leq&
\left(|K_{b,(m)}(\bs{\omega})|
\int\! p_{t-1}(\bs{x}_{t-1})\,d\bs{x}_{t-1}\!\right)^{\!\!2}
\!=|K_{b,(m)}(\bs{\omega})|^2.
\end{eqnarray*}
Thus,
\begin{equation}
\int ||\bs{\omega}||^{2\beta}
|\F[\,\ov{p}_{t,i_1,\dots,i_m}^{(m)}]|^2\,d\bs{\omega}
\leq
\int ||\bs{\omega}||^{2\beta} 
|K^{(m)}_b(\bs{\omega})|^2\leq (2\pi)^d L^2_{K_{b,(m)}}.
\label{derPt1}
\end{equation}

The above formula shows that 
$\ov{p}_{t,i_1,\dots,i_m}^{(m)}\in\mathcal{P}_{(\beta,L_{K_{b,(m)}})}$
for any $t\in\N$. We proceed with the specification of the Sobolev 
constant $L_{t,(m)}$ of the partial derivative $p_{t,i_1,\dots,i_m}^{(m)}$.

Similarly as in the proof of Theorem~\ref{Th2}, we have
$\sup_{\bs{x}_t,\bs{y}_t} \{|g_t(\bs{y}_t|\bs{x}_t)|\}=||g_t||_{\infty}<\infty$.
Further, multiplying the update formula (\ref{UpdateFT}) 
by the complex exponential, integrating, multiplying by
$(-\mathrm{i})^m(\omega^{i_1}_{i_1}\cdot\dots\cdot\omega^{i_d}_{i_d})$
and the respective conjugates we shift to the Fourier transforms of
partial derivatives and get
\begin{eqnarray*}
(\ov{\pi}_tg_t)\,p_t(\bs{x}_t)&=&g_t(\bs{x}_t)\,\ov{p}_t(\bs{x}_t),\\
(\ov{\pi}_tg_t)\!\!
\int\!\! e^{\mathrm{i}\lss\bs{\omega},\bs{x}_t\rss}p_t(\bs{x}_t)\,d\bs{x}_t
&=&
\!\!\int\!\! e^{\mathrm{i}\lss\bs{\omega},\bs{x}_t\rss}
g_t(\bs{x}_t)\,\ov{p}_t(\bs{x}_t)\,d\bs{x}_t,\\
(\ov{\pi}_tg_t)\,\psi_t(\bs{\omega})
&\leq&||g_t||_{\infty}^2\,\ov{\psi}_t(\bs{\omega}),\\
(\ov{\pi}_tg_t)\,
(-\mathrm{i})^m\omega^{i_1}_{i_1}\cdot\dots\cdot\omega^{i_d}_{i_d}\,
\psi_t(\bs{\omega})
&\leq&||g_t||_{\infty}^2\,
(-\mathrm{i})^m\omega^{i_1}_{i_1}\cdot\dots\cdot\omega^{i_d}_{i_d}\,
\ov{\psi}_t(\bs{\omega}),\\
(\ov{\pi}_tg_t)^2\,|\F[\,p_{t,i_1,\dots,i_m}^{(m)}]|^2
&\leq& ||g_t||_{\infty}^2\,|\F[\,\ov{p}_{t,i_1,\dots,i_m}^{(m)}]|^2,\\
||\bs{\omega}||^{2\beta}\,|\F[\,p_{t,i_1,\dots,i_m}^{(m)}]|^2
&\leq& \frac{||g_t||_{\infty}^2}{(\ov{\pi}_tg_t)^2}\,||\bs{\omega}||^{2\beta}
|\F[\,\ov{p}_{t,i_1,\dots,i_m}^{(m)}]|^2,\nonumber\\
(2\pi)^{-d}\!\int\! ||\bs{\omega}||^{2\beta}
|\F[\,p_{t,i_1,\dots,i_m}^{(m)}]|^2\,d\bs{\omega}
&\leq&
\frac{||g_t||_{\infty}^2L^2_{K_{b,(m)}}}{(\ov{\pi}_tg_t)^2}=L^2_{t,(m)}.
\label{derLtdef}
\end{eqnarray*}
This concludes the proof.\hfill$\Box$

\section{Example}
\label{SecVI}
In this section, we demonstrate an application of the presented theory. 
Because our research has not been driven by any concrete application, 
we apply the particle filtering and kernel density estimation methodologies 
on the filtering problem for a~multivariate Gaussian process. This problem
has the analytical solution - the well-known Kalman filter \cite{Kalman1960, 
Krylov2007, DSG1999, Sarkka2013}. 

The purpose of this choice is to check if empirical results from computer
simulations follow the analytic counterpart. By replacing the Gaussian
transition kernel and Gaussian observation density by general entities
we can build up the appropriate particle filter for a general Markov process,
but without the possibility of checking against the analytical solution.

\subsection{Multivariate Gaussian process}
Let the signal and observation processes introduced in Section~\ref{SigSec}
be specified as multivariate Gaussian. That is, we assume that the formulas
driving evolution of states and observations are specified, for a
general dimension $d\geq 1$,~as
\begin{equation}
\label{mvgprocess}
\bs{X}_{t}=\bsm{F}\bs{X}_{t-1}+\bs{W}_t,\;\;\;
\bs{Y}_t=\bsm{H}\bs{X}_t+\bs{V}_t,\;\;\;t\in\N,
\end{equation}
where $\bsm{F}$, $\bsm{H}$ are $d\times d$ regular matrices and
$W_t\sim\mathcal{N}(\bs{0},\mathrm{\bsm{Q}})$,
$V_t\sim\mathcal{N}(\bs{0},\mathrm{\bsm{R}})$ are multivariate
normal noise terms with $d\times d$ covariance matrices $\mathrm{\bsm{Q}}$
and $\mathrm{\bsm{R}}$. The signal process $\{\bs{X}_t\}_{t=0}^\infty$ forms
a multivariate Markov chain with Gaussian transition kernels. The initial
distribution is considered also multivariate normal, i.e., 
$\bs{X}_0\sim\mathcal{N}(\bs{\mu}_0,\mathrm{\bsm{\Sigma}_0})$,
$\bs{\mu}_0\in\R^d$ and $\bsm{\Sigma}_0$ is a $d\times d$ covariance matrix.

Mathematically, the filtering task is to find the conditional expected values 
$\E[\bs{X}_t|\bs{Y}_1,\dots,\bs{Y}_t]$ for $t\geq 1$. At the given
time instant~$t\in\N$, the conditional expected value is the integral
characteristic of the related conditional distribution which
represents the filtering distribution we are interested~in.

The vector $(\bs{X}_0,\bs{X}_1,\bs{Y}_1,\dots,\bs{X}_t,\bs{Y}_t)$
is multivariate normal because it is determined by a linear transformation
of the vector 
$(\bs{X}_0,\bs{W}_1,\bs{V}_1,...,\bs{W}_t,\bs{V}_t)$
which is multivariate normal.
Therefore, the filtering distribution is also multivariate normal, 
and is determined by its mean vector $\bs{\mu}_t$ and 
its covariance matrix~$\bs{\Sigma}_t$ at time $t\in\N$.
The preservation of the normal character of the filtering distribution
over time allows us to obtain an analytic expression for its parameters. 
The result is known as the \textit{multivariate Kalman filter}.

\subsection{Multivariate Kalman filter}
The theoretical analysis presented in \cite{DSG1999}
gives the following recursive Kalman's equations for $\bs{\mu}_t$
and $\bs{\Sigma}_t$. The parameters are computed
in several steps using some auxiliary variables for $t\geq 1$:
\begin{eqnarray*}
\bs{\mu}_{t|t-1}&=&\bsm{F}\bs{\mu}_{t-1},
\nonumber\\
\bs{\Sigma}_{t|t-1}&=&\bsm{F}\bs{\Sigma}_{t-1}\bsm{F}^T+\bsm{Q},
\nonumber\\
\bsm{K}_t&=&\bs{\Sigma}_{t|t-1}\bsm{H}^T[\bsm{H}\bs{\Sigma}_{t|t-1}\bsm{H}^T+\bsm{R}]^{-1},
\nonumber\\
\bs{\mu}_t&=&\bs{\mu}_{t|t-1}+\bsm{K}_t[\bs{Y}_t-\bsm{H}\bs{\mu}_{t|t-1}],\\
\label{mvK1}
\bs{\Sigma}_{t}&=&[\bsm{I}_d-\bsm{K}_t\bsm{H}]\bs{\Sigma}_{t|t-1}.
\label{mvK2}
\end{eqnarray*}

Using the above formulas, one can recursively compute the determining
parameters of the filtering distribution over time. Due to the normal
character of the distribution, we have apparently 
$\E[\bs{X}_t|\bs{Y}_1,\dots,\bs{Y}_t]=\bs{\mu}_t$.
Further, the formula for the evolution of the covariance matrix 
$\bs{\Sigma}_{t}$ is deterministic. That is, it is not affected by 
observations.

\subsection{Multivariate Gaussian particle filter}
The incorporation of schema (\ref{mvgprocess}) into the particle 
filter's computation, presented in Section~\ref{SMCAlgSec},
stems from the specification of the initial density $p_0(\bs{x}_0)$ and
the set of transition kernels $K_{t-1}$, $t\in\N$.

As already mentioned, the initial density is multivariate
normal with some mean $\bs{\mu}_0\in\R^d$ and a $d\times d$ 
covariance matrix $\bs{\Sigma}_0$, i.e.,
$$
\label{mvp}
p_0(\bs{x}_0)=
(2\pi)^{-\frac{d}{2}}
|\bs{\Sigma}_0|^{-\frac{1}{2}}
\!\exp\!\left[
-\frac{1}{2}(\bs{x}_0-\bs{\mu}_0)^T\bs{\Sigma}^{-1}\!(\bs{x}_0-\bs{\mu})
\right]\!.
$$
The densities of Gaussian transition kernels $K_{t-1}(\bs{x}_t|\bs{x}_{t-1})$, 
$t\in\N$ are specified~as
\begin{equation}
\label{mGKeq}
K_{t-1}(\bs{x}_t|\bs{x}_{t-1})=(2\pi)^{-\frac{d}{2}}
|\bsm{Q}|^{-\frac{1}{2}}
\exp\left[-\bs{u}_t^T\bsm{Q}^{-1}\bs{u}_t\right]
\end{equation}
with $\bs{u}_t=\bs{x}_t-\bsm{F}\bs{x}_{t-1}$.

The above formula reflects the multivariate normal character of the noise term
$\bs{W}_t$ in (\ref{mvgprocess}) and, in fact, corresponds to the 
specification of the density of the multivariate normal distribution 
$\mathcal{N}(\bsm{F}\bs{x}_{t-1},\bsm{Q})$.

The Sobolev character of the filtering densities is given by the Sobolev character 
of the Gaussian transition kernels. We show that the conditional characteristic
functions of the Gaussian kernels (\ref{mGKeq}) are uniformly bounded, 
which implies the Sobolev character according to Theorem~\ref{Th2}. 

We have 
$\F[K_{t-1}](\bs{\omega}|\bs{x}_{t-1})=\F[\mathcal{N}(\bsm{F}\bs{x}_{t-1},\bsm{Q})]$,
and therefore
$$
\F[K_{t-1}](\bs{\omega}|\bs{x}_{t-1})=
e^{\mathrm{i}\lss\bs{\omega},\bsm{F}\bs{x}_{t-1}\rss}
\exp\left[-\frac{1}{2}\bs{\omega}^{T}\bsm{Q}\,\bs{\omega}\right].
$$
Further,
\begin{eqnarray*}
|\F[K_{t-1}](\bs{\omega}|\bs{x}_{t-1})|
&\leq&
\exp\left[-\frac{1}{2}\bs{\omega}^{T}\bsm{Q}\,\bs{\omega}\right]\\
&\leq&
\exp\left[-\frac{1}{2}\lambda_{\min}||\bs{\omega}||^2\right]
=K_b(\bs{\omega}),
\end{eqnarray*}
where $\lambda_{\min}$ is the minimal eigenvalue of the covariance
matrix~$\bsm{Q}$.

For the Sobolev constant $L_{K_b}$ of $K_b(\bs{\omega})$ and $\beta=1$,
we have the integral
$$
(2\pi)^{-d}\!\!
\int ||\bs{\omega}||^{2}\exp\left[-\lambda_{\min}||\bs{\omega}||^2\right]
=\frac{\pi^{-\frac{d}{2}}}{4^d(\sqrt{\lambda_{\min}})^{d+2}}.
$$

From this result we also see that any multivariate normal initial distribution
with the covariance matrix $\bsm{\Sigma}_0$ is $1$-Sobolev with the constant
$L_0=\pi^{-\frac{d}{2}}/[4^d(\sqrt{\lambda^0_{\min}})^{d+2}]$,
where $\lambda^0_{\min}$ is the minimal eigenvalue of $\bsm{\Sigma}_0$.

The obtained result on the Sobolev character of the filtering densities is consistent with
the fact that all densities in the multivariate Gaussian process (\ref{mvgprocess}) 
are normal, i.e., the character of the involved densities does not change during
operation of the filter.

\subsection{Multivariate Gaussian convolution kernel}
In the multivariate Gaussian particle filter, kernel density 
estimates are made using the multivariate standard normal 
(convolution) kernel
$$
K(\bs{u})=
(2\pi)^{-\frac{d}{2}}
\exp\left[-\frac{1}{2}||\bs{u}||^2\right].
$$

The specification of the $L_2$ norm of the kernel is straightforward.
We have
\begin{eqnarray*}
||K||^2
&=&
(2\pi)^{-d}\!\!
\int \exp(-||\bs{u}||^2)\,d\bs{u}
=(4\pi)^{-\frac{d}{2}},
\end{eqnarray*}
hence $||K||=(4\pi)^{-\frac{d}{4}}$.

Concerning the $A$ constant of Theorem \ref{bdTheorem},
we start with the Fourier transform
of the multivariate standard normal kernel which corresponds to the
characteristic function of the $\mathcal{N}(\bs{0},\bs{\mathrm{I}}_d)$ distribution.
That is,
$K_{\F}(\bs{\omega})=e^{-\bs{\omega}^T\bs{\mathrm{I}}_d\bs{\omega}}
=e^{-\frac{1}{2}||\bs{\omega}||^2}$. In order to specify some constant $A$,
we need to determine a bound on the spectral matrix norm of the Hessian of $K_{\F}$.
The entries of the Hessian matrix $\mathcal{H}(K_{\mathcal{F}})$ reads as
$$
\frac{\partial^2 K_{\mathcal{F}}}{\partial\omega^2_j}=
(\omega_j^2-1)K_{\mathcal{F}}(\bs{\omega}),\;
\frac{\partial K_{\mathcal{F}}}{\partial \omega_j\partial \omega_k}=
K_{\mathcal{F}}(\bs{\omega})\omega_j\omega_k,\;\;j\not=k.
$$
In the matrix notation, the Hessian writes as
$\mathcal{H}(K_{\mathcal{F}})(\bs{\omega})=
K_{\mathcal{F}}(\bs{\omega})(\bs{\omega}\bs{\omega}^T-\bs{\mathrm{I}}_d)$.
Using the spectral matrix norm we get 
\begin{eqnarray*}
||\mathcal{H}(K_{\mathcal{F}})(\bs{\omega})||_{spc}&\leq&
K_{\mathcal{F}}(\bs{\omega})
||\bs{\omega}\bs{\omega}^T-\bs{\mathrm{I}}_d||_{spc}\\
&\leq&
K_{\mathcal{F}}(\bs{\omega})
(||\bs{\omega}\bs{\omega}^T||_{spc}+||\bs{\mathrm{I}}_d||_{spc})\\
&\leq&
K_{\mathcal{F}}(\bs{\omega})
(||\bs{\omega}^T||||\bs{\omega}||+1)\\
&\leq&
K_{\mathcal{F}}(\bs{\omega})
(||\bs{\omega}||^2+1).
\end{eqnarray*}

Note that for a vector $\bs{\omega}\in\R^d$,
$||\bs{\omega}||_{spc}=||\bs{\omega}||$ (the standard Euclidean norm).
Let $\bs{\omega}=\bs{\xi}$ such that $||\bs{\xi}||\leq 1$. Then we clearly
have $||\mathcal{H}(K_{\mathcal{F}})(\bs{\xi})||_{spc}\leq 2$ as
$K_{\mathcal{\F}}(\bs{\xi})\leq 1$.

The multidimensional Taylor's theorem for $K_{\F}$ writes as
$$
K_{\F}(\bs{\omega})
=K_{\F}(\bs{0})+
(\nabla K_{\F}(\bs{0}))\bs{\omega}+
\frac{1}{2}
\bs{\bs{\omega}}^{T}[\mathcal{H}(K_{\mathcal{F}})(\bs{\xi})]\bs{\omega}
$$
for a suitable $\bs{\xi}\in\R^d$, $||\bs{\xi}||\leq ||\bs{\omega}||$.
For the gradient, we have $\nabla K_{\F}(\bs{0})=\bs{0}$ and $K_{\F}(\bs{0})=1$,
therefore the above Taylor's theorem gives for any $||\bs{\omega}||\leq 1$,
\begin{eqnarray*}
K_{\F}(\bs{\omega})-1&=&
\frac{1}{2}
\bs{\omega}^{T}[\mathcal{H}(K_{\mathcal{F}})(\bs{\xi})]\bs{\omega},\\
|K_{\F}(\bs{\omega})-1|&\leq&
\frac{1}{2}
||\bs{\omega}^T||\cdot
||\mathcal{H}(K_{\mathcal{F}})(\bs{\xi})||_{spc}\cdot
||\bs{\omega}||,\\
\frac{|K_{\F}(\bs{\omega})-1|}{||\bs{\omega}||}
&\leq&||\bs{\omega}^T||=||\bs{\omega}||.
\end{eqnarray*}

Further $|K_{\F}(\bs{\omega})-1|\leq 1$ for all $\bs{\omega}\in\R^d$ and
therefore $|K_{\F}(\bs{\omega})-1|/||\bs{\omega}||\leq 1$ for
$||\bs{\omega}||>1$. Thus, joining the two inequalities we finally get
$$
\frac{|K_{\F}(\bs{\omega})-1|}{||\bs{\omega}||}
\leq \max\{1,1\}=1,\;\;\bs{\omega}\in\R^d\backslash\{\bs{0}\},
$$
and the $A$ constant equals to 1, i.e., $A=1$.

The above considerations immediately lead to the specification of 
the order of the multivariate standard normal kernel.
As mentioned, the Fourier transform of the kernel is 
$K_{\F}(\bs{\omega})=e^{-\frac{1}{2}||\bs{\omega}||^2}$ and
$K_{\F}(\bs{0})=\bs{1}$. The related gradient writes as
$\nabla K_{\F}(\bs{\omega})=-e^{-\frac{1}{2}||\bs{\omega}||^2}\bs{\omega}$,
thus $\nabla K_{\F}(\bs{0})=\bs{0}$. For the Hessian of $K_{\F}(\bs{\omega})$,
we have $\mathrm{diag}(\mathcal{H}(K_{\F})(\bs{0}))=-\bs{1}$.
Hence the order of the kernel is $\ell=\beta=1$.

\subsection{MATLAB implementation and experiments}
In this section we introduce our implementation of the multivariate
Kalman filter and its particle filter counterpart to show results 
of several experiments.

We have implemented both filters in the form of a MATLAB function. 
The inputs into the function are $\bsm{F},\bsm{Q},\bsm{H},\bsm{R}$
matrices of formula (\ref{mvgprocess}), the computational horizon
$T\in\N$ and the selected number of particles $n\in\N$. The outputs
are the means and covariance matrices from the particle and Kalman
filters, respectively. If the dimension of the signal process is $d=1$
or $d=2$, then the script provides a graphical output illustrating
the estimated density and its theoretical counterpart from the Kalman filter.
The source code of the function is presented in Appendix~\ref{mvsmc.m}.

We have performed several experiments in order to check if the 
computational behavior of the multivariate Gaussian particle filter 
coincides with the analytical results. The experiments were performed
for the following setting of parameters:
$\bsm{F}=\bsm{I}_d$, $\bsm{Q}=2\bsm{I}_d$, $\bsm{H}=2\bsm{I}_d$, 
$\bsm{R}=\bsm{I}_d$. In the script, the density of the
multivariate standard normal distribution is used as the initial density.
Computational horizon was set to $T=100$.

The results of three $d=2$ experiments for different numbers of particles
$n=10,100$ and $n=1000$ are presented in Table~\ref{resMKF}. 
Graphically, the obtained kernel density estimate and theoretical 
filtering density are presented in Fig.~\ref{mv100} for $n=100$.

\begin{table}[!htb]
\begin{center}
\begin{tabular}{l|cc|cc|cc}
$T$=100&$\hat{\bs{\mu}}_T$&$\bs{\mu}_T$&
\multicolumn{2}{c|}{$\widehat{\bsm{\Sigma}}_T$\;-\;PF}&
\multicolumn{2}{c}{$\bsm{\Sigma}_T$\;-\;KF}\\\hline\hline
$n$=10&\;32.25&\;31.92&\;\;0.1472&\;\;0.0992&0.2247&0\\
&-18.43&-18.65&\;\;0.0992&\;\;0.4290&0&0.2247\\\hline
$n$=100&\;\;\,0.46&\;\;\,0.48&\;\,0.1557&-0.0212 &0.2247&0\\
&\;\,-2.16&\;\,-2.04&\,-0.0212&\;0.2144&0&0.2247\\\hline
$n$=1000&\;\,-2.76&\;\,-2.75&\;\,0.2207&-0.0036&0.2247&0\\
&-29.18&-29.18&\,-0.0036&\;0.2206&0&0.2247\\
\end{tabular}
\end{center}
\caption{Comparison of bivariate particle and Kalman filters.}
\label{resMKF}
\end{table}

\begin{figure}[!htb]
\centerline{\includegraphics[width=4.8in]{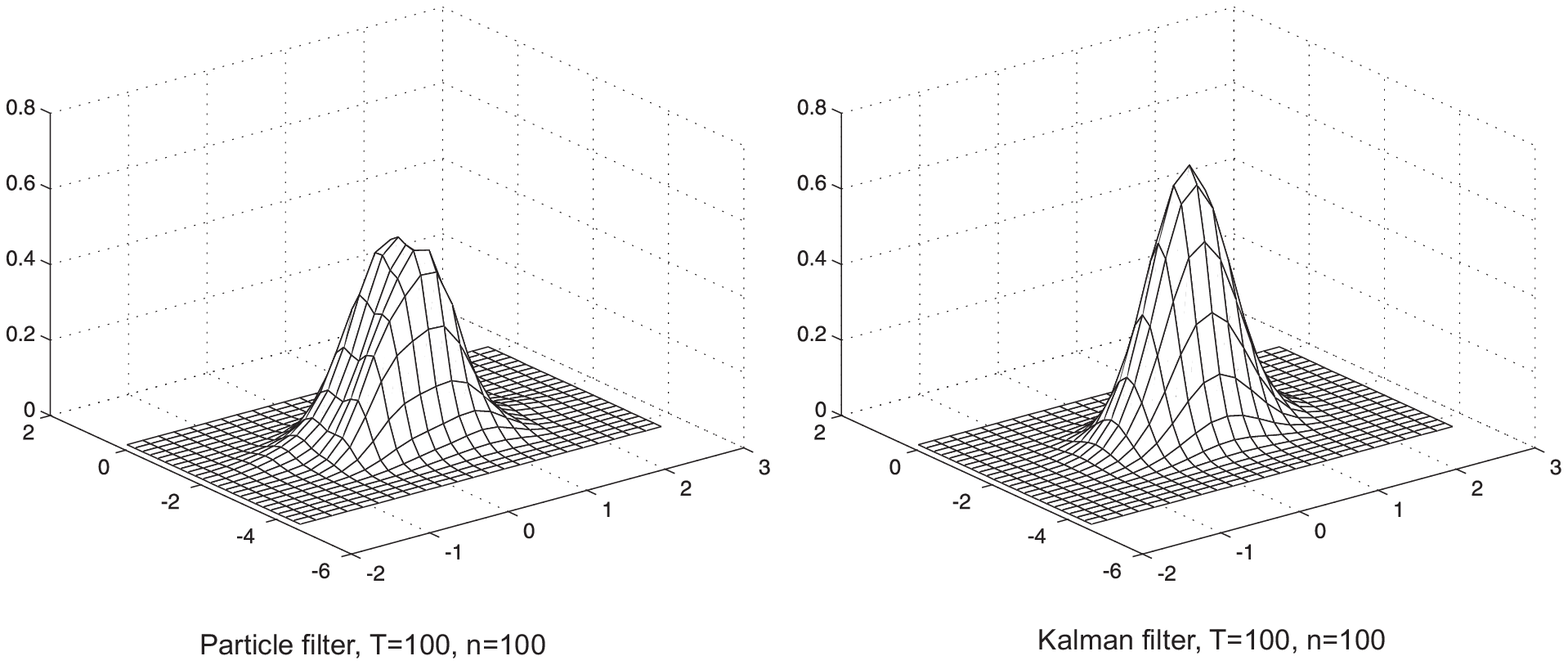}}
\caption{The kernel density estimate generated by the bivariate Gaussian 
particle filter and the corresponding filtering density from 
the Kalman filter.}
\label{mv100}
\end{figure}

On the basis of the inspection of the numerical results presented in 
Table~\ref{resMKF}, we can state a good agreement of numerical
characteristics delivered by the Gaussian particle filter with the
theoretical characteristics of the filtering distributions.

\section{Conclusion}
\label{SecVII}
In the paper, we have demonstrated that the standard methodology
of kernel density estimates can be applied in the area of particle
filtering. We have proved that the kernel density estimates constructed
on the basis of particles generated by the particle filter converge in
the MISE to the theoretical filtering density at each time instant
of operation of the filter. The result holds even though
the generated particles do not constitute an i.i.d. sample from the
filtering distribution. Moreover, we have stated the sufficient condition
for the preservation of the Sobolev character of the filtering densities
over time. The extension of both results to the partial derivatives
of the kernel estimates and filtering densities has been provided as well.

In Theorem 2, the constant $c_t $ is known that it typically grows
exponentially with time, see e.g., \cite{Doucet2001} p.~87, therefore $C_t$
of (\ref{MThEq2}) does so; and, if one wants to assure the given
precision of the density approximation, then one must increase
the number of generated particles exponentially, too. 
This is an unpleasant property of the particle filter. On the other hand,
there are results available, e.g., \cite{DelMoral2001} or \cite{Heine2008}, that
under additional conditions, uniformly convergent particle filters 
can be constructed. That is, that $c_t$ of (\ref{ctformula})
is constant over time.

The constant $C_t$ depends on $L_t$. Under the conditions
of Theorem~\ref{Th2}, we know the evolution of $L_t$ over time. 
In fact, the evolution is somehow similar to the evolution of $c_t$
constant and there is again the risk of an exponential growth of $L_t$.
The study of the conditions when $L_t$ evolves uniformly over
time is the issue of the future research in this field.



%

\appendix
\section{MATLAB implementation}
\label{mvsmc.m}
{\small
\begin{verbatim}
function [PFm,PFcov,KFm,KFcov] = mvpf(F,Q,H,R,T,n);

%---HMM---
d=size(Q,1);
m0=zeros(d,1);S0=eye(d);
X0=mvnrnd(m0',S0)';
X=zeros(d,T);Y=X;
X(:,1)=F*X0+mvnrnd(zeros(1,d),Q)';
Y(:,1)=H*X(:,1)+mvnrnd(zeros(1,d),R)';
for t=2:T,
 W=mvnrnd(zeros(1,d),Q)';
 V=mvnrnd(zeros(1,d),R)';
 X(:,t)=F*X(:,t-1)+W;
 Y(:,t)=H*X(:,t)+V;
end;

%---Kalman filter---
M=zeros(d,T);
m=m0;S=S0;
for t=1:T,
 m1=F*m;
 S1=F*S*F'+Q;
 K=S1*H'*inv(H*S1*H'+R);
 m=m1+K*(Y(:,t)-H*m1);
 S=(eye(d)-K*H)*S1;
 M(:,t)=m;
end;
KFm=M(:,T)
KFcov=S;

%---PF filter---
P=mvnrnd(m0',S0,n)';
for t=1:T,
 disp(t);   
 pp=zeros(d,n);w=zeros(1,n);
 for j=1:n;
  pp(:,j)=F*P(:,j)+mvnrnd(zeros(1,d),Q)';
  w(j)=mvnpdf((Y(:,t)-H*pp(:,j))',zeros(1,d),R);
 end;
 if n>1, wn=w/sum(w); else wn=1; end;
 mn=randsample(n,n,true,wn);
 P=pp(:,mn);
end;
PT=P;
PFm=mean(PT')';
PFcov=cov(PT');

%---kernel estimate for d=1 with graphical output---
if d==1,
 alpha=1;beta=1;
 hn=alpha*n^(-1/(2*beta+1));
 xx=[KFm-5*sqrt(KFcov):0.1:KFm+5*sqrt(KFcov)];
 fx=zeros(1,length(xx));
 for j=1:n,
  fx=fx+1/(n*hn)*1/sqrt(2*pi)*exp(-(xx-PT(j)).^2/(2*hn^2));
 end;
 plot(xx,fx,'b',xx,normpdf(xx,KFm,sqrt(KFcov)),'r');
 figure(1);
end;

%---kernel estimate for d=2 with graphical output---
if d==2,
 alpha=1;beta=1;
 hn=alpha*n^(-1/(2*beta+d));
 x1=[KFm(1)-5*sqrt(KFcov(1,1)):0.2:KFm(1)+5*sqrt(KFcov(1,1))];
 x2=[KFm(2)-5*sqrt(KFcov(2,2)):0.2:KFm(2)+5*sqrt(KFcov(2,2))];
 [X1,X2]=meshgrid(x1,x2);
 Xr=[X1(:) X2(:)];
 nXr=size(Xr,1);fxr=zeros(nXr,1);
 for j=1:n,
  mvn=mvnpdf((Xr-ones(nXr,1)*PT(:,j)')/hn,zeros(1,d),eye(d));
  fxr=fxr+1/(n*hn^d)*mvn;
 end;
 colormap([0 0 0]);
 mesh(X1,X2,reshape(fxr,length(x2),length(x1)));
 figure(1);
 pause;
 pr=mvnpdf(Xr,KFm',KFcov);
 mesh(X1,X2,reshape(pr,length(x2),length(x1)));
 figure(1);
end;
\end{verbatim}}

\section*{Acknowledgment}
The author is grateful to V. Bene\v{s} (Charles University in Prague) for stimulating interest
in the field of particle filtering. The research was supported by COST grant
LD13002 provided by the Ministry of Education, Youth and Sports of the Czech Republic.

\end{document}